\documentclass[11pt]{article}

\RequirePackage{etex}
\usepackage[flushleft]{threeparttable}

\usepackage[margin=1in]{geometry}   
\usepackage{amsfonts}
\usepackage{amsmath} 
\usepackage{amssymb}
\usepackage{centernot}
\usepackage{booktabs}
\usepackage{multirow}
\usepackage{algorithm}% http://ctan.org/pkg/algorithms
\usepackage{algorithmicx}% http://ctan.org/pkg/algorithmic

\usepackage[onehalfspacing]{setspace} 
\usepackage{float}
\usepackage{caption}
\usepackage{graphicx}
\usepackage{adjustbox}
\usepackage{lscape}

\usepackage{adjustbox}

\usepackage{xcolor}

\usepackage{relsize}

\usepackage{bbm}
\usepackage{subfig}

\definecolor{forestgreen}{RGB}{34,139,34}
\usepackage{hyperref}
\hypersetup{
    colorlinks=true,
    linkcolor=magenta,
    filecolor=magenta, 
    citecolor=forestgreen,      
    urlcolor=blue
}

\usepackage{authblk}

\usepackage{amsthm}
\newtheorem{theorem}{Theorem}

\usepackage{xpatch}
\makeatletter
\xpatchcmd{\proof}{\@addpunct{.}}{\@addpunct{:}}{}{}
\makeatother

\makeatletter
\def\@hangfrom#1{\setbox\@tempboxa\hbox{{#1}}%
      \hangindent 0pt%\wd\@tempboxa
      \noindent\box\@tempboxa}
\makeatother

\makeatletter
\newcommand{\vast}{\bBigg@{3}}
\newcommand{\Vast}{\bBigg@{4}}
\makeatother

\makeatletter
\newcommand*{\indep}{%
  \mathbin{%
    \mathpalette{\@indep}{}%
  }%
}
\newcommand*{\nindep}{%
  \mathbin{%                   % The final symbol is a binary math operator
    \mathpalette{\@indep}{\not}% \mathpalette helps for the adaptation
  }%
}
\newcommand*{\@indep}[2]{%
  \sbox0{$#1\perp\m@th$}%        box 0 contains \perp symbol
  \sbox2{$#1=$}%                 box 2 for the height of =
  \sbox4{$#1\vcenter{}$}%        box 4 for the height of the math axis
  \rlap{\copy0}%                 first \perp
  \dimen@=\dimexpr\ht2-\ht4-.2pt\relax
  \kern\dimen@
  {#2}%
  \kern\dimen@
  \copy0 %                       second \perp
} 
\makeatother

\DeclareMathOperator{\E}{\textnormal{\mbox{E}}}

\usepackage[utf8]{inputenc}
\DeclareUnicodeCharacter{200E}{}

\usepackage{tikz}
\usetikzlibrary{positioning,shapes.geometric}

\usepackage{setspace}
\usepackage{cite}

\usepackage[stable]{footmisc}

\usepackage{rotating}

\makeatletter
\def\@seccntformat#1{\@ifundefined{#1@cntformat}%
   {\csname the#1\endcsname\quad}  % default
   {\csname #1@cntformat\endcsname}% enable individual control
}
\let\oldappendix\appendix %% save current definition of \appendix
\renewcommand\appendix{%
    \oldappendix
    \newcommand{\section@cntformat}{\appendixname~\thesection\quad}
}
\makeatother

\usepackage[absolute,showboxes]{textpos}

\TPMargin{5pt}

\usepackage{thmtools, thm-restate}
\usepackage{datetime}

\usepackage{sectsty}
\subsectionfont{\noindent\textbf\itshape}
\subsubsectionfont{\normalfont\itshape}

\begin{document}

\title{Interpretable meta-analysis of model or marker performance}

\author[1]{Jon A. Steingrimsson}
\author[2,3]{Lan Wen}
\author[1]{Sarah Voter}
\author[3-5]{Issa J. Dahabreh}

\affil[1]{Department of Biostatistics, Brown University School of Public Health, Providence, RI }
\affil[2]{ Department of Statistics and Actuarial Science, University of Waterloo, Waterloo, Ontario, Canada}
\affil[3]{CAUSALab, Harvard T.H. Chan School of Public Health, Boston, MA}
\affil[4]{Department of Epidemiology, Harvard T.H. Chan School of Public Health, Boston, MA}
\affil[5]{Department of Biostatistics, Harvard T.H. Chan School of Public Health, Boston, MA}

\maketitle{}

\clearpage

\vspace*{1in}

\begin{abstract}
\noindent
\linespread{1.3}\selectfont
Conventional meta analysis of model performance conducted using datasources from different underlying populations often result in estimates that cannot be interpreted in the context of a well defined target population. In this manuscript we develop methods for meta-analysis of several measures of model performance that are interpretable in the context of a well defined target population when the populations underlying the datasources used in the meta analysis are heterogeneous. This includes developing identifiablity conditions, inverse-weighting, outcome model, and doubly robust estimator. We illustrate the methods using simulations and data from two large lung cancer screening trials.
\end{abstract}

\clearpage

\section{Introduction}

Users of prediction models are often interested in using model derived predictions in a target population of substantive interest. When multiple studies evaluate the performance of the same prediction model (or biomarker), stakeholders are often interested in synthesizing evidence across the studies to learn about model performance in a target population of substantive interest (i.e.~performing meta analysis). Commonly used databases for meta analysis of model development and evaluation are randomized controlled trials, observational cohort studies, and electronic health records. But the populations underlying each datasource likely differ and they are usually not a random sample of the target population. For example, participants who agree to participate in clinical research are often younger \cite{hutchins1999underrepresentation}, healthier \cite{pinsky2007evidence}, and less diverse \cite{national2022improving} than the population that meet the eligibility criteria for participation, and observational cohort studies and administrative databases are often limited to specific healthcare systems or geographic regions.

Meta analysis of measures of model performance is usually conducted using weighted averages of study specific estimators of model performance or by modeling the distribution of measures of model performance across studies \cite{debray2017guide, debray2019framework}. But when prediction error modifiers, that is, variables that are associated with prediction error for a given measure of model performance and prediction model \cite{steingrimsson2021transporting}, are differentially distributed across the studies or between the target population and the study populations, estimates of measures of model performance obtained from standard meta analysis methods are not reflective of measures of model performance in the target population (and usually don't have a clear interpretation outside of the population underlying the pooled sample from the studies). The analogous problem occurs for meta-analysis of treatment effects and methods have been developed for synthesizing data from multiple randomized controlled trials that have a causal interpretation in a target population \cite{dahabreh2020toward, dahabreh2023efficient}. 

Several methods have been developed for transporting or generalizing both loss-based measures of model performance \cite{shimodaira2000improving, sugiyama2007covariate, sugiyama2012machine,morrison2023robust,sahoo2022learning,angelopoulos2023prediction, ge2023maximum} and estimators of area under the curve from a single study to a target population \cite{li2022estimating}.  The related literature on multi-source domain adaptation \cite{zhang2015multi,sun2015survey, zhao2020multi,nomura2021efficient,mansour2008domain, qiu2023efficient} usually focuses on risk estimation using samples from multiple source populations that perform well in a target population or are robust to distributional shifts between the source populations and the target population.

In this manuscript we develop identifiability results and propose novel estimators for measures of model performance in a target population, using data from multiple studies. We refer to this approach as interpretable meta analysis of model performance. The measures of model performance we consider are risk-based measures, sensitivity, specificity, negative predictive value, positive predictive value, and the area under the curve. We derive asymptotic properties of the estimators and evaluate finite sample properties of the estimators using simulations and apply them to analyzing data on lung cancer screening.

\section{Data and target parameters}

We assume that we have data from $K$ studies that we refer to as the source studies. For each source study participant, we assume that information is available on which study they come from $S = 1, \ldots, K$, covariate information $X$, and outcome information $Y$, where the outcome can be binary, count, or continuous. That is, the observed data from study $s$ is
\[
\{(X_i, S_i =s,  Y_i), 1, \ldots, n_s\},
\]
where $n_s$ is the number of observations in study $s \in \{1, \ldots, K\}$. We also assume that we have covariate information, but no outcome information, on a sample from the target population. Let $S=0$ denote an observation being from the target population and $n_0$ be the number of observations available from the target population. We denote $\mathcal{O}$ as the combined dataset from all the studies and the target population with $n = \sum_{k=0}^K n_k$ representing the total number of observations in the combined dataset. We assume that the observations are independent and identically distributed random variables, and define $R = I(S \neq 0)$ as an indicator of whether an observation comes from one of the studies ($R=1$) or the target population ($R=S=0$).

Throughout, we make the following two assumptions:
\begin{enumerate}
\item[A1] Conditional exchangeability between datasources. We assume that $Y \indep R|X$.
\item[A2] Positivity of participation in the collection of data sources: $\Pr[R =1|X = x] > 0$ for all $x$ such that $f(x|S=0)>0$.
\end{enumerate}
The positivity assumption A2 states that all covariate patterns that can occur in the target population have a positive probability of occurring in at least one of the studies. Hence, the methods developed can draw conclusions about a target population that has a broader spectrum than each of the individual studies.

%The conditional exchangeability assumption is a fairly strong assumption that implies:
%\begin{equation*}
%\E[Y|X, S=1] = \ldots = \E[Y|X, S=K] = \E[Y|X, R=1] = \E[Y|X, S=0].
%\end{equation*}
%and the conditional exchangeability assumption also has testable implications as
%\begin{equation}
%\label{test-imp}
%\E[Y|X, S=1] = \ldots = \E[Y|X, S=K] = \E[Y|X, R=1],
%\end{equation}
%which only depends on the observed data. Testing \ref{test-imp} can be done using tests for equivalence of conditional expectations \cite{luedtke2019omnibus,racine2006testing}.

Let $h(X, \widehat \beta)$ be a prediction model indexed through a parameter $\beta$ and let $h(X, \widehat \beta)$ be the estimated model where the parameter $\widehat \beta$ is an estimator for $\beta$. Throughout we do not make the assumption that the model is correctly specified. As such, the results presented in the rest of this paper apply to both correctly specified and misspecified models. We are interested in estimating performance of the model $h(X, \widehat \beta)$ on a dataset that is independent of the data used for model building (i.e.,~the data used to calculate $\widehat \beta$). 

For a binary outcome, common measures of model performance, especially when estimating diagnostic accuracy, are sensitivity and specificity. For a cut-point $c$, sensitivity in the target population is defined as
\[
\Pr[h(X, \widehat \beta) > c|Y = 1, S=0]
\]
and specificity in the target population is defined as
\[
\Pr[h(X, \widehat \beta) \leq c|Y = 0, S=0].
\]
That is, sensitivity in the target population is the proportion of the target population that have the disease ($Y=1$) that the model classifies as having the disease $(h(X, \widehat \beta) > c)$. And, specificity is the proportion of the target population that are disease free ($Y=0$) that the model classifies as disease free $(h(X, \widehat \beta) \leq  c)$. 

Both sensitivity and specificity condition on true disease status while positive predictive value and negative predictive value condition on the model derived classification. The positive predictive value in the target population is defined as $\Pr[Y=1|I(h(X, \widehat \beta) > c), S=0]$ and the negative predictive value in the target population is defined as $\Pr[Y=0|I(h(X, \widehat \beta) \leq c),S=0]$.

Plotting sensitivity vs $1-$specificity when varying the threshold $c$ gives the receiver operating characteristic (ROC) curve. The area under the ROC curve (AUC) provides a summary of sensitivity and specificity across thresholds. The AUC can be interpreted as the probability that a randomly selected observation that has the disease ($Y=1$) will have a higher model derived risk than a randomly selected observation without the disease ($Y=0$). Mathematically, the AUC in the target population is defined as
\[
\E[I(h(X_i, \widehat \beta) > h(X_j, \widehat \beta))|Y_i =1, Y_j = 0, S_i =0, S_j = 0]
\]
where the indices $i$ and $j$ denote random observations from the target population with $Y=1$ and with $Y=0$, respectively.

Loss-based measures of model performance such as the mean squared error, absolute loss, and the Brier score \cite{brier1950verification} are another class of commonly used measures. A loss function $L(Y, h(X, \widehat \beta))$ measures the discrepancy between the outcome $(Y)$ and the model derived predictions $(h(X, \widehat \beta))$. The risk (expected loss) in the target population is defined by $\E[L(Y, h(X, \widehat \beta))|S=0]$. In the main part of the manuscript we focus on estimating sensitivity and AUC and mostly present results for specificity, negative predictive value, positive predictive value, and loss-based measures of model performance in the Supplementary Web Appendix.

\section{Identifiability of measures of model performance}
\label{sec:id}

The following theorem shows that the sensitivity in the target population is identifiable using the observable data $\mathcal{O}$. A proof is given in Supplementary Web Appendix \ref{app:proofs}.
\begin{theorem}
\label{thm:id}
If assumptions A1 and A2 hold and $\E[\Pr[Y=1|X,R=1]|S=0]>0$, then the sensitivity in the target population is identifiable using the observed data through the observed data functional 
\begin{equation}
\label{id:out}
\frac{\E[I(h(X, \widehat \beta) > c) \Pr[Y = 1|X, R=1]|S=0]}{\E[\Pr[Y=1|X,R=1]S=0]}
\end{equation}
or equivalently using the weighting representation
\begin{equation}
\label{id:ipw}
\frac{\E\left[\frac{\Pr[R=0|X]}{\Pr[R=1|X]}I(h(X, \widehat \beta) > c, Y=1, R=1) \right]}{\E\left[\frac{\Pr[R=0|X]}{\Pr[R=1|X]}I(Y=1,R=1)\right]}.
\end{equation}
\end{theorem}
In Supplementary Web Appendix \ref{app:proofs} we prove the following identifiability result for the AUC.
\begin{theorem}
\label{thm-id-auc}
Assume that assumptions A1 and A2 hold and $\normalfont\E[\Pr[Y=1| X_i,R=1] (1 - \Pr[Y=1| X_j,R=1])|S_i=0, S_j=0] >0$, where $i$ is a random observation from the target population with $Y=1$ and $j$ is a random observation from the target population with $Y=0$. Then the AUC in the target population is identifiable using the observed data functional
\begin{equation}
\label{AUC-ID-O}
\frac{ {\normalfont\E} [\Pr[Y_k=1|R_k=1,X_k] \Pr[Y_l=0|R_l=1,X_l] I(h( X_k,  \widehat \beta) > h( X_l,  \widehat \beta))| S_k=0, S_l=0]}{{\normalfont\E}[\Pr[Y_k=1|R_k=1,X_k] \Pr[Y_l=0|R_l=1,X_l]|S_k=0, S_l=0]},  
\end{equation}
or equivalently using the weighting representation
\begin{equation} 
\label{AUC-ID-W}
\frac{{\normalfont\E}\left[\frac{\Pr[R=0|X_k]}{\Pr[R=1|X_k]}\frac{\Pr[R=0|X_l]}{\Pr[R=1|X_l]} I(h( X_k,  \widehat \beta) > h( X_l,  \widehat  \beta), Y_k=1, Y_l=0, R_k=1, R_l=1)\right]}{{\normalfont\E}\left[\frac{\Pr[R=0|X_k]}{\Pr[R=1|X_k]}\frac{\Pr[R=0|X_l]}{\Pr[R=1|X_l]}  I(Y_k=1, Y_l=0, R_k=1, R_l=1)\right]}.
\end{equation}
Here, $k$ is a random observation from the target population with $Y=1$ and $l$ is a random observation from the target population with $Y=0$.
\end{theorem}
In Supplementary Web Appendix \ref{app:NPV} we provide identifiability results and estimators for loss-based measures of model performance, specificity, and negative and positive predictive value.

\section{Sampling framework} 

We adopt the sampling framework from causally interpretable meta-analysis for treatment effects \cite{dahabreh2023efficient}. That is, we assume that there is a single population stratified by $S$ and we assume that the data from each of the $K$ studies can be modeled as being a sample from some superpopulation that can be potentially ill-defined and hard to characterize. We also assume that we have access to a random sample from a well defined target population. We allow the sampling fractions for the studies and the target population to be unknown and potentially unequal. More specifically, we assume there are two sampling models i) a population sampling model which assumes we have a random sample from the single superpopulation and ii) a biased sampling model where the sampling is done stratified on $S$. If we denote densities from the population sampling model using $p(\cdot)$ and densities from the biased sampling model using $q(\cdot)$, then the assumptions made imply  $p(y|x, r=1) = q(y|x, r=1)$ and $p(x|r=0) = q(x|r=0)$. And as the probabilities in the identifiability expressions \eqref{id:out} and \eqref{AUC-ID-O} only rely on the densities $p(y|x, r=1)$ and $p(x|r=0)$, they can be interpreted as being derived from the biased sampling model, with the expectations and probabilities integrated with respect to densities from the biased sampling model. As the biased sampling model is a more realistic representation of how the data is collected, we will work under the biased sampling setting.

\section{Estimation of measures of model performance in the target population}
\label{sec:est}

Sample analogs of expression \eqref{id:out} give the following outcome (or g-formula like \cite{robins1986new}) estimator for sensitivity in the target population:
\begin{equation}
\label{sens-out}
\widehat \psi_{sens, out} = \frac{\sum_{i=1}^n I(R_i = 0) I(h(X_i, \widehat \beta) > c) \widehat m(X_i)}{\sum_{i=1}^n I(R_i=0)  \widehat m(X_i)},
\end{equation}
where $\widehat m(X)$ is an estimator for $\Pr[Y = 1|X, R=1]$.

The sample analogs of expression \eqref{id:ipw} gives the following weighting estimator:
\begin{equation}
\label{sens-w}
\widehat \psi_{sens, w} =\frac{\sum_{i=1}^n I(h(X_i, \widehat \beta) > c, Y_i=1, R_i =1) \widehat w(X_i)}{\sum_{i=1}^n I(Y_i=1, R_i =1) \widehat w(X_i)},
\end{equation}
where $\widehat w(X)$ is an estimator for $\frac{\Pr[R=0|X]}{\Pr[R=1|X]}$. In Supplementary Web Appendix \ref{app:send-inf} we derive the non-parametric influence function for sensitivity in the target population. This influence function based ``doubly robust" estimator is given by
\[
\widehat \psi_{sens, dr} = \frac{\sum_{i=1}^n \left(I(S_i=0) I(h(X_i, \widehat \beta) > c) \widehat m(X_i) +  \widehat w(X_i) I(R_i=1) I(h(X_i, \widehat \beta) > c) \big\{ I(Y_i = 1) - \widehat m(X_i)\big\}\right)}{\sum_{i=1}^n \left(I(S_i=0) \widehat m(X_i) +  \widehat w(X_i) I(R_i=1) \big\{ I(Y_i = 1) - \widehat m(X_i)\big\}\right)}
\]
In Supplementary Web Appendix \ref{app:asym-prop} we proof the following theorem.
\begin{theorem}
\label{thm-sens-asymp}
 If either $\widehat m(X) \overset{P}{\longrightarrow}\Pr[Y = 1|X, R=1]$ or $\widehat w(X) \overset{P}{\longrightarrow}\frac{\Pr[R=0|X]}{\Pr[R=1|X]}$, then $\widehat \psi_{sens, dr}$ is consistent (i.e.,~$\widehat \psi_{sens, dr} \overset{P}{\longrightarrow} \psi_{sens}$).
\end{theorem}
So far we have assumed that the cut-point $c$ is provided, but it is often data-dependent. For example, the Youden index \cite{youden1950index} is a way to select cut-points by maximizing the sum of sensitivity and specificity. Using the methods developed, the Youden index can easily be extended to perform cut-point selection in the target population. Let $\widehat \psi_{spec, out}(c)$ be the outcome model estimator for specificity in the target population provided in the Supplementary Web Appendix, then the outcome model estimator for the optimal cut-off point is $\widehat c_{out} = \underset{c}{\max} \hspace{2pt} (\widehat \psi_{sens, out}(c) + \widehat \psi_{spec, out}(c) -1)$ and the analogous versions using weighting or doubly robust estimators can also be used. Here, for each evaluation measure the notation $(c)$ denotes that the evaluation measure is calculated using the cut-point $c$.

We use the notation in \cite{li2022estimating} to define estimators for the AUC in the target population. For a function $k(X_i,X_j)$, define 
$d^{out} (O_i, O_j; k(X_i,X_j)) = \widehat m(X_i) (1-\widehat m(X_j))   I(S_i=0, S_j=0) k(X_i,X_j).$ The outcome model estimator for the area under the curve in the target population is given by
\begin{equation} \label{g-auc}  
\widehat \psi_{auc, out}  = \frac{\sum_{i\neq j} d^{out}(O_i, O_j; I(h(X_i, \widehat \beta) > h(X_j, \widehat \beta))}{\sum_{i\neq j} d^{out}(O_i, O_j; 1)}.
\end{equation}
Similarly, for a function $k(X_i,X_j)$, define $d^{w}(O_i, O_j; k(X_i,X_j)) = \widehat w(X_i) \widehat w(X_j) I(Y_i=1, Y_j=0, R_i=1, R_j=1) k(X_i,X_j).$
The weighting-based estimator for the AUC in the target population is given by
\begin{equation} \label{IOW-AUC}  
\widehat \psi_{auc, w}  = \frac{\sum_{i\neq j} d^{w}(O_i, O_j; I(h(X_i, \widehat \beta) > h(X_j, \widehat \beta))}{\sum_{i\neq j} d^{w}(O_i, O_j; 1)}.
\end{equation}
Lastly, we define the doubly robust estimator for the AUC in the target population as
\begin{equation} \label{dr}
\widehat \psi_{auc,dr}
= \frac{ \sum_{i \neq j}  d^{dr} (O_i, O_j; k(X_i,X_j) = I(h(X_i, \widehat \beta) > h(X_j, \widehat \beta))}{  \sum_{i \neq j} d^{dr}(O_i, O_j; k(X_i,X_j)= 1)},
\end{equation}
where 
\begin{equation} \label{sums}
\begin{split}d^{dr}(O_i, O_j; k(X_i,X_j)) &= d^{w}(O_i, O_j; k(X_i,X_j)) +  d^{out} (O_i, O_j; k(X_i,X_j)) -\\
&\widehat w(X_i) \widehat w(X_j) \widehat m(X_i) (1 - \widehat m(X_j))   I(R_i=1, R_j=1) k(X_i,X_j).
 \end{split}
\end{equation}
The following theorem shows that the doubly robust estimator for the AUC in the target population is doubly robust and a proof is given in the Supplementary Web Appendix.
\begin{theorem}
\label{thm-auc-asymp}
 For a random variable $W$ define $\mathbb{G}_n(W) = \sqrt{n} \left(\frac{1}{n} \sum_{i=1}^n W_i - E[W]\right)$. We assume that the models $\widehat m(X)$ and $\widehat w(X)$ are parametric and indexed by parameters $\theta_1$ and $\theta_2$, respectively, and that $\left\{\mathbb{G}_n (m(X;\theta_1)): \theta_1 \in \Theta_1\right\}$ and $\left\{\mathbb{G}_n (m(X;\theta_2)): \theta_2 \in \Theta_2\right\}$ are stochastically equicontinuous where $\Theta_j$ is the parameter space for $\theta_j, j=1,2$. If either $\widehat m(X) \overset{P}{\longrightarrow}\Pr[Y = 1|X, R=1]$ or $\widehat w(X) \overset{P}{\longrightarrow}\frac{\Pr[R=0|X]}{\Pr[R=1|X]}$, then $\widehat \psi_{auc, dr}$ is consistent (i.e.,~$\widehat \psi_{auc, dr} \overset{P}{\longrightarrow} \psi_{auc}$).
\end{theorem}

\section{Relaxing the conditional exchangeability assumption using knowledge of the marginal prevalence rate in the target population}

So far we have assumed that the conditional exchangeability assumption $Y \indep R|X$ holds. Now assume that this assumption is violated ($ Y \mathlarger{\nindep} R | X$) and therefore for at least one $s' \in \{1, \ldots, K\}$ we have $f_{Y | X, S}(y | x,  s = 0) \neq f_{Y | X,S}(y | x, s = s')$. To relax that assumption we assume for each study that the relationship between the population underlying the study sample and the target population can be expressed through an exponential tilt model \cite{scharfstein2018globalBiometrics, scharfstein2018globalSMMR,scharfstein2021global, dahabreh2022global}
\begin{equation}\label{model_exponential_tilt}
	f_{Y | X, S}(y | x, s = 0) \propto e^{ \gamma_{s'} y} f_{Y | X,  S}(y| x , s = s'), s' \in \{1, \ldots, K\}
\end{equation} 
for some sensitivity parameter $\gamma_{s'} \in \mathbb R$, for $s' \in \{1, \ldots, K\}$. As outcome information is unavailable in the target population, $\gamma_{s'}$ cannot be estimated using the observed data and hence needs to be specified a priori. If $\gamma_{s'} = 0$ for all $s' \in \{1, \ldots, K\}$, then it can be shown that the  conditional exchangeability assumption $Y \indep R|X$ holds. 

As $f_{Y | X, S}(y | x, s = 0)$ is a density, it follows that for all $s' \in \{1, \ldots, K\}$ 
\begin{equation*}
  \begin{split}
  f_{Y | X, S}( y | x, 0) %&= \dfrac{f_{Y^a | X, S}(y| x, 1) e^{\psi_a y}}{\mathlarger\int\limits_{y^\prime \in \mathcal Y^a} f_{Y^a| X, S}( y^\prime | x, s') e^{\psi_a q(y^\prime)}dy^\prime} \\
      &= \dfrac{ e^{\gamma_{s'} y} f_{Y | X, S}(y| x, s') }{ \E [ e^{\gamma_{s'} y} | X = x, S = s']}
  \end{split}
\end{equation*} 
which implies
\begin{equation}
\label{test-sens}
\dfrac{ e^{\gamma_{1} y} f_{Y | X, S}(y| x, 1) }{ \E [ e^{\gamma_{1} y} | X = x, S = 1]} =\dfrac{ e^{\gamma_{2} y} f_{Y | X, S}(y| x, 2) }{ \E [ e^{\gamma_{2} y} | X = x, S = 2]} = \ldots = \dfrac{ e^{\gamma_{K} y} f_{Y | X, S}(y| x, K) }{ \E [ e^{\gamma_{K} y} | X = x, S = K]}.
\end{equation}
As outcome information is available in all the studies, all quantities appearing in the equation above only rely on the observed data and hence is testable (e.g.,~using test of equality of distributions \cite{luedtke2019omnibus}). However, rejecting the null hypothesis that all the densities in equation \eqref{test-sens} are equal does not indicate which, if any, of the studies satisfies equation \eqref{model_exponential_tilt}. And, if the test of equality of densities in Equation \eqref{test-sens} does not provide evidence against the equality of densities in equation \eqref{test-sens}, and we believe that Equation \eqref{test-sens} holds, that does not imply that assumption of the exponential tilt model \eqref{model_exponential_tilt} is satisfied.

In Supplementary Web Appendix \ref{app:sens} we proof the following theorem.
\begin{theorem}
Under the exponential tilt model, the sensitivity in the target population when transporting from study $s'$ is identifiable using the observed data functional 
\begin{equation}
\label{sens-sens}
 \psi_{sens, out}(\gamma_{s'}) = \frac{\E \left[\frac{\E[ I(h(X, \widehat \beta) > c, Y=1) e^{\gamma_{s'} y}|X, S=s']}{\E[e^{\gamma_{s'} y}|X, S=s']} \Bigg|S=0\right]}{\E \left[\frac{\E[ I(Y=1) e^{\gamma_{s'} y}|X, S=s']}{\E[e^{\gamma_{s'} y}|X, S=s']} \Bigg|S=0\right]}.
\end{equation}
\end{theorem}
Expression \eqref{sens-sens} suggests the following outcome model estimator for the sensitivity in the target population
\begin{equation}
\label{est-sens-sens}
\widehat \psi_{sens, out}(\gamma_{s'}) = \frac{\sum_{i=1}^n I(S_i=0) \widehat b_{s'}^1(X_i)}{\sum_{i=1}^n I(S_i=0) \widehat b_{s'}^0(X_i)},
\end{equation}
where $\widehat b_{s'}^j(X)$ is an estimator for
$\frac{\E[ I(h(X, \widehat \beta) > c)^j I(Y=1) e^{\gamma_{s'} y}|X, S=s']}{\E[e^{\gamma_{s'} y}|X, S=s']}$
that can be implemented by $\widehat b_{s'}^j(X) =  \dfrac{  I(h(X, \widehat \beta) > c)^j  e^{\gamma} \widehat m_{s'}(X) }{  1 + \widehat m_{s'}(X) (e^{\gamma} - 1)}$ for $j=\{0,1\}$,
where $\widehat m_{s'}(X)$ is an estimator for $\Pr[Y=1|X, S= s']$.

To implement $\widehat \psi_{sens, out}(\gamma_{s'})$ the sensitivity analysis parameter $(\gamma_{s'})$ needs to be selected and as outcome information is unavailable in the target population it cannot be estimated using the data. However if information on the marginal prevalence rate in the target population $E[Y|S=0]$ is available, then under the exponential tilt model
\begin{equation}
\label{eq-sens-sel}
  \int \int y  \dfrac{ e^{\gamma_{s'} y} f_{Y | X, S}(y| x, s') }{ \E [ e^{\gamma_{s'} y} | X = x, S = s']} f_{X|S}(x | 0) dy dx - E[Y|S=0] = 0
\end{equation}
and $\gamma_{s'}$ can be estimated as a solution to the sample analog of equation \eqref{eq-sens-sel}.

The estimators $\widehat \psi_{sens, out}(\gamma_{s'}), s' \in \{1, \ldots, K\}$ transport sensitivity at the individual study level to the same target population and are therefore estimating the same target parameter, provided that all relevant assumptions hold. Incompatibility of the data with the asymptotic equality of all $\widehat \psi_{sens, out}(\gamma_{s'})$ suggests that at least one assumption is violated (i.e.,~at least one of the estimators $\widehat \psi_{sens, out}(\gamma_{k})$ is biased but we cannot from the data determine which it is). If the investigator believes that all assumptions hold and the data does not provide evidence to the contrary, a natural question is how to combine $\widehat \psi_{sens, out}(\gamma_{s'}), s' \in \{1, \ldots, K\}$ to improve efficiency. Any linear combination $\sum_{i=1}^K \widehat a_k \widehat \psi_{sens, out}(\gamma_{k})$ is a consistent estimator for the sensitivity in the target population as long as the weights $\widehat a_k, k = 1, \ldots, K$ sum to one (here, the notation $\widehat a_k$ is used to denote that the weights can be data dependent). Hence, the main consideration when choosing the weights is efficiency. In standard meta analysis, selecting the weights $\widehat a_k$ proportional to the inverse variance of each estimator $\widehat \psi_{sens, out}(\gamma_{k})$ minimizes asymptotic variance \cite{zeng2015random}. However, this result does not hold when combining $\widehat \psi_{sens, out}(\gamma_{k})$ as they all use the same target population data and are therefore not independent \cite{steingrimsson2023systematically}. For each $s' \in \{1, \ldots, K\}$, equation \eqref{est-sens-sens} puts restrictions on the joint distribution of $(X,S,S \times Y)$ and as $\widehat \psi_{sens, out}(\gamma_{k})$ are estimated using different parts of the data there is no guarantee that in finite samples there exists a joint distribution that is compatible with all the estimators. However if all the assumptions hold and all the nuisance functions are correctly specified, then asymptotically the estimators are compatible \cite{robins2000sensitivity,dahabreh2022global}.

\section{Simulations}

To evaluate the finite sample performance of the estimators, we performed simulations comparing the outcome model estimator, the weighting estimator, and an estimator that estimates model performance using pooled data from the source studies, referred to as the source estimator. 

We simulated the covariate vector from a five dimensional multivariate normal distribution with mean zero and a covariance matrix where element $(i,j)$ is equal to $0.6^{|i-j|}$. The selection into any study ($R$) was simulated from a logistic regression model $R \sim Ber(\Pr[R=1|X])$ where
\[
\Pr[R=1|X] = expit(1 + 0.5 X_1 + 0.5 X_2 + 0.3 X_3 + 0.3 X_1^2 + 0.3 X_2^2 + 0.3 X_3^2),
\]
where $expit(x) = \frac{e^x}{1 + e^{x}}$. The covariate dependent selection into the target population ensures that there are differences between the covariate distributions in the target population and the population underlying the pooled data from all the studies ($R=1$). We considered a setting where there are three studies ($K=3$) and study participants were selected into one of the three studies by simulating from a multinomial logistic regression model $S|X, R=1 \sim Multinomial((p_1, p_2, p_3), \sum_{i=1}^n R_i)$, where $p_1 = \frac{\beta}{1 + \beta + \eta}$, $p_2 = \frac{\eta}{1 + \beta + \eta}$, $p_3 = 1 - p_1 - p_3$, with $\beta = exp(log(1.3) X_1 + log(1.3)X_2 + log(1.3) X_3 )$ and $\eta = exp(log(0.8) X_1 + log(0.8) X_2 + log(0.8) X_3)$. The covariate dependent selection into each study ensures that the populations underlying each study have different covariate distributions. The outcome was simulated from a logistic regression model $Y \sim Ber(\Pr[Y=1|X])$ where
\[
\Pr[Y=1|X] = expit(1 + 0.5 X_1 + 0.2 X_2 + 0.3 X_1^2 + 0.3 X_2^2).
\]
The model we evaluate the performance of was the asymptotic limit (estimated numerically) of a logistic regression model in the source population that used linear main affects of all five covariates. As that model does not include $X_1^2$ and $X_2^2$, it is misspecified. To evaluate the performance of the model in the target population we used sensitivity, specificity, negative predictive value, positive predictive value, the Brier score, and the area under the curve. The cut-point required to define sensitivity, specificity, negative and positive predictive value was selected using the Youden index (maximizing the sum of sensitivity and specificity) from the observations with $R=1$. 

We ran $1000$ simulations where the total sample size was $n = 2000$. Using this set-up, the prevalence rate of the outcome is approximately 75\% in the target population and 82\% in the collection of the study data. The average sample size in the data from the target population is $366$, $565$ in study one, $573$ in study two, and $496$ in study three.

The simulation results are presented in Figure \ref{fig:mis-spec}. The results show that the outcome model, doubly robust, and the weighting estimators  were nearly unbiased for all measures of model performance. The source estimator is biased for all the measures ($6-200\%$ relative bias) as it fails to account for the differences between the target population and the population underlying the pooled data from all the studies. In Supplementary Web Appendix \ref{app-sim} we present simulation results under correctly specified models and for cases where the outcome model and the model for study participation were estimated using a generalized additive models. The results from these additiona simulation studies show the same trends to those observed in Figure \ref{fig:mis-spec}.

\begin{figure}[htbp]
    \centering
    \includegraphics[width = \textwidth]{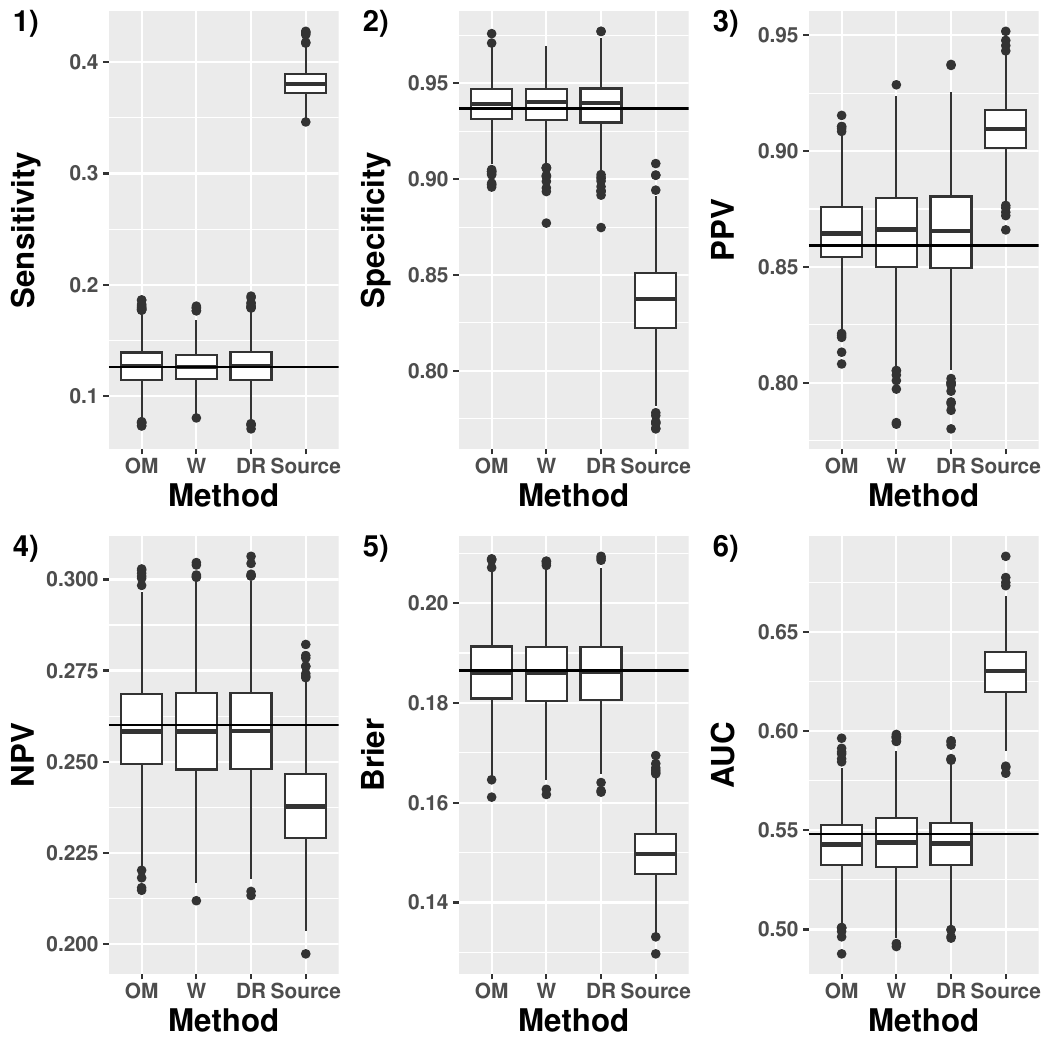}
    \caption{\label{fig:mis-spec} Simulation results for estimating model performance in the target population estimated using the outcome model (OM) estimator, weighting (W) estimator, doubly robust (DR) and an estimator that uses only data from the collection of studies (Source). The measures of model performance evaluated are: sensitivity (Plot 1), specificity (Plot 2), positive predictive value (Plot 3), negative predictive value (Plot 4), Brier score (Plot 5), and area under the curve (Plot 6). The horizontal line is the true measure of model performance (calculated analytically). Here, the model that is evaluated is correctly specified.}
\end{figure}

\section{Analysis using data on lung cancer screening}

In this section we apply the estimators developed to data on lung cancer screening. We use data from two large lung cancer screening trials, the National Lung Screening Trial (NLST) \cite{national2011reduced} and the Prostate, Lung, Colorectal, and Ovarian (PLCO) Cancer Screening Trial \cite{oken2011screening}. We are interested in evaluating the performance of a model for estimating the risk of being diagnosed with lung cancer within six years from study entry. 

NLST randomized participants to either screening with chest X-ray or screening with low-dose computed tomography. Participants enrolled to PLCO were randomized to either an usual care arm or an arm where participants received screening for various cancers including chest X-ray screening for lung cancer. As both trials include screening with chest X-ray we focus on the subset of participants that were randomized to receive chest X-ray screening for lung cancer. The outcome of interest is whether a participant was diagnosed with lung cancer within six years from study enrollment (binary outcome). 

A natural target population for lung cancer screening is everyone in the US who is eligible for lung cancer screening. To avoid violations of the positivity assumption and as the criteria used by U.S. Preventive Services Task Force to recommend people for lung cancer screening \cite{krist2021screening} are very similar to the eligibility criteria for NLST, we define the target population as all people in the US who meet the eligibility criteria for the NLST (people aged 55 to 74 that had $\geq 30$ pack-year history who were current smokers or had quit within the past 15 years). To obtain a reasonable sample from the target population we used data from the 2003-2004 NHANES survey. NHANES is a cross-sectional survey constructed to be representative of the non-institutionalized US population. We used the subset of the NHANES participants that participated in a smoking sub-study and met the NLST eligibility criteria (data from the smoking sub-study is needed to assess the eligibility criteria). Each observation in NHANES is associated with a sampling weight that accounts for, among other things, oversampling of certain subgroups and survey non-respondence. In Section \ref{sec:weights} in the Supplementary Web Appendix we show how the estimators presented in Section \ref{sec:est} can be modified to account for sampling weights.

For simplicity, and due to the limited amount of missing data, we focus on the participants that have complete data on covariates listed in Table \ref{tab:DA} and for NLST and PLCO participants that have six years of follow-up. This resulted in $22,920$ participants used from NLST, $17,639$ from PLCO, and $219$ participants from NHANES (representing approximately $8.5$ million participants). 

We split the NLST data into equal sized training and test sets. Using the data from the training set from NLST we built a logistic regression model with linear main effects of all variables listed in Table \ref{tab:DA}. To evaluate the performance of that model in the target population, we used data from the test set of the NLST and PLCO data as the study data and data from NHANES as the sample from the target population. To select the cut-point (c) used to calculate sensitivity, specificity, NPV, and PPV we used the Youden index calculated using data from the NLST training set (i.e.,~the value that maximizes the sum of sensitivity and specificity). To estimate the nuisance parameters needed for the implementation of the outcome model and weighting estimators logistic regression, models with linear main effects of all variables listed in Table \ref{tab:DA} were used. 

To calculate confidence intervals we used a stratified bootstrap procedure that was consistent with the NHANES sampling design \cite{shao2003impact}. We started by resampling at the first primary sampling unit level and within each first primary sampling unit resampled at the level of the secondary sampling unit. To estimate confidence intervals for the source estimator we used the non-parametric bootstrap. We note that there is some potential for non-identifiable overlap between the samples from the NLST, PLCO, and NHANES. But considering the large number of people eligible for lung cancer screening, substantial overlap between the samples is unlikely.

Table \ref{tab:DA} shows summary statistics for the variables used in the analysis stratified by datasource. The table shows that the summary statistics are relatively similar between the NLST and PLCO while there are substantial differences between the NHANES and the two trials (e.g., the NHANES sample has more comorbidities, less education, and is more racially diverse).

\begin{table}[htbp] 
\small
\centering 
\caption{\small Summary of participant characteristics for participants in the National Lung Screening Trial (NLST), the Prostate, Lung, Colorectal, and Ovarian (PLCO) Cancer Screening Trial, and the NHANES 2003-2004 survey. Continuous variables are summarized using mean (standard deviation) and categorical variables are summarized by percentages in each category. The summaries for the NHANES data are weighted by the NHANES sampling weight. BMI is body max index; Smoke years is the total number of years the participant smoked cigarettes; Smoke age is the age at smoking onset; Pack years is calculated as (Total number of years Smoked × Cigarettes Per Day/20).}
\label{tab:DA}
\label{res-data-an}\label{tab-1}
\begin{threeparttable}[htbp]
\begin{tabular}{@{}llll@{}}
\toprule
Variable                                                                & NLST        & PLCO        & NHANES     \\ \midrule 
Age                                                                     & 61.3 (5.0)  & 62.4 (5.2)  & 63.1 (5.5)  \\ 
BMI                                                                     & 27.9 (5.1)  & 27.6 (4.8)  & 28.6 (5.6)  \\ 
Race (White)                                                            &    90.8\%   & 91.0\%      & 84.5\%      \\ 
Education level & & & \\
\multicolumn{1}{l}{Some college education}                                                        &   55.5\%    & 53.2\%      & 45.7\%      \\ 
\multicolumn{1}{l}{High school graduate}                                                        &   38.8\%    & 37.5\%      & 29.1\%      \\ 
Smoke years                                                              &  39.6 (7.3) & 36.1 (9.2) & 42.5 (7.4)  \\ 
Gender (Male)                                                           &    58.7\%   & 63.8\%      & 63.0\%      \\ 
Marital status (Married)                                                &   68.4\%    & 72.8\%      & 64.7\%      \\ 
Pack year                                                               & 55.5 (23.3) & 57.0 (26.2) & 60.6 (28.9) \\ 
History of diabetes (Yes)                                                               &   9.3\%    & 8.5\%       & 20.3\%      \\ 
History of emphysema (Yes)                                                              &    7.3\%    & 6.5\%       & 8.7\%       \\ 
\begin{tabular}[c]{@{}l@{}}History of heart disease \\ or heart attack (Yes)\end{tabular} &    12.2\%   & 13.2\%      & 28.0\%      \\ 
History of hypertension (Yes)                                                           &   35.2\%   & 35.3\%      & 45.8\%      \\ 
Cigarettes per day categorical & & & \\
\multicolumn{1}{l}{11-20}                                                     & 47.6\% & 36.5\% &  53.8\% \\  
\multicolumn{1}{l}{21-30}                                                     & 27.4\% & 30.7\% &  19.2\% \\  
\multicolumn{1}{l}{31-40}                                                     & 18.0\% & 19.4\% &  16.8\% \\  
\multicolumn{1}{l}{41-60}                                                     & 6.2\% & 10.9\% &  9.0\% \\ 
\multicolumn{1}{l}{61-80}                                                     & 0.6\% & 2.1\% & 1.1\%  \\ 
\multicolumn{1}{l}{$>80$}                                                     & 0.1\% & 0.4\% &  0\%  \\
\bottomrule
\end{tabular}
\end{threeparttable}
\end{table}
Table \ref{res-data-an} shows estimates and standard errors from the outcome model estimator, the doubly robust estimator, the weighting estimator, and an estimator that pools data from the NLST test set and the PLCO trial for AUC, sensitivity, specificity, PPV, NPV, and Brier score. The result shows that the outcome model, doubly robust, and weighting estimators are similar for all measures suggesting limited impact of the specification of the outcome model and the model for study participation. The pooled estimator has lower Brier score, lower sensitivity, and higher specificity compared to the transportability estimators while the AUC and the negative and positive predictive values are similar to the outcome model and weighting estimators.

\begin{table}[htbp] 
\small
\centering 
\caption{\small Estimates and standard errors from the outcome model, weighting, and source population estimators of area under the curve, sensitivity, specificity, positive predictive value (PPV), negative predictive value (NPV), and Brier risk. Outcome model refers to estimates from the outcome model estimator, Doubly Robust refers to estimates from the doubly robust estimator, weighting refers to estimates from the weighing estimator, and source refers to estimates from an estimator that combines data from the NLST and PLCO.}
\label{res-data-an}\label{tab-1}
\begin{threeparttable}[htbp]
\begin{tabular}{@{}lllll@{}}
\toprule
        &  Outcome Model & Weighting & Doubly Robust & Source \\ \midrule
Area under the curve & 0.704 (0.013)                                                             &  0.680 (0.011)   & 0.689 (0.016)                                                      & 0.700 (0.0064)     \\ 
    
Sensitivity & 0.842 (0.029)                                                             & 0.834 (0.024)                         & 0.841 (0.029)                              & 0.613 (0.012)     \\ 
Specificity & 0.409 (0.046)                                                             & 0.402 (0.034)            & 0.408 (0.046)                                             & 0.673 (0.029)     \\ 
PPV         & 0.114 (0.0083)                                                            & 0.106 (0.0052)             & 0.111 (0.0078)                                           & 0.0971 (0.0029)   \\ 
NPV         & 0.966 (0.0026)                                                            & 0.966 (0.0020)          & 0.967 (0.0029)                                              & 0.968 (0.0013)    \\ 
Brier risk & 0.0835 (0.0051)                                                           & 0.0719 (0.0038)                   & 0.0824 (0.0049)                                    & 0.0543 (0.0012)   \\ 
\bottomrule
\end{tabular}
\end{threeparttable}
\end{table}

\section{Discussion}

In this manuscript we developed methods for meta analysis of prediction models that can be interpreted in the context of a target population. We provided identifiability results and estimators for several measures of model performance. The estimators show good performance in simulations and are used to analyse data on lung cancer screening. For the lung cancer screening analyses the data from the target population is associated with a sampling weight and we extend the estimators to handle sampling weights. All of the developments hold when evaluating the performance of both correctly specified and misspecified models.

In our set-up, the model that is being evaluated is built on a dataset that is independent of the data used for model development. This incorporates when an external evaluation is done, when a split into a training and a test set is used, the use of cross-validated measures of model performance, and when evaluating the performance of a biomarker. 

The methods developed assume that we have access to individual level data on participants from all studies and individual level data from the target population. Developing methods that can be used with summary level information would be useful. Other extensions of interest include extensions to more complex data-structures, systematically-missing data \cite{steingrimsson2022systematically}, and addressing measurement error.

\section*{Acknowledgements}

This work was supported in part by National Library of Medicine (NLM) Award R01LM013616, and Patient-Centered Outcomes Research Institute (PCORI) awards ME-2019C3-17875, and ME-2021C2-22365. Statements in this paper do not necessarily represent the views of the PCORI, its Board of Governors, the Methodology Committee, or the NIH. We thank the National Cancer Institute (NCI) for access to the National Lung Screening Trial (NLST) and PCLO data. This paper does not necessarily reflect the opinions or views of NCI, NHLBI, PLCO, or NLST. 
L. Wen is supported by the Natural Sciences and Engineering Research Council of Canada (NSERC) Discovery Grant [RGPIN-2023-03641, DGECR-2023-00455].

%%%%%%%%%%%%%%%%%%%%%%%%%%%%%%%%%%%%%%%%%%%%%%%%%%%%%%%%%%%%%%%%%%%%%%%%%%%%%%
% BIBLIOGRAPHY
%%%%%%%%%%%%%%%%%%%%%%%%%%%%%%%%%%%%%%%%%%%%%%%%%%%%%%%%%%%%%%%%%%%%%%%%%%%%%%
\bibliographystyle{ieeetr}
\bibliography{references}

\begin{thebibliography}{10}

\bibitem{hutchins1999underrepresentation}
L.~F. Hutchins, J.~M. Unger, J.~J. Crowley, C.~A. Coltman~Jr, and K.~S. Albain,
  ``Underrepresentation of patients 65 years of age or older in
  cancer-treatment trials,'' {\em New England Journal of Medicine}, vol.~341,
  no.~27, pp.~2061--2067, 1999.

\bibitem{pinsky2007evidence}
P.~Pinsky, A.~Miller, B.~Kramer, T.~Church, D.~Reding, P.~Prorok, E.~Gelmann,
  R.~Schoen, S.~Buys, R.~Hayes, {\em et~al.}, ``Evidence of a healthy volunteer
  effect in the prostate, lung, colorectal, and ovarian cancer screening
  trial,'' {\em American journal of epidemiology}, vol.~165, no.~8,
  pp.~874--881, 2007.

\bibitem{national2022improving}
N.~A. of~Sciences~Engineering, Medicine, {\em et~al.}, ``Improving
  representation in clinical trials and research: building research equity for
  women and underrepresented groups,'' 2022.

\bibitem{debray2017guide}
T.~P. Debray, J.~A. Damen, K.~I. Snell, J.~Ensor, L.~Hooft, J.~B. Reitsma,
  R.~D. Riley, and K.~G. Moons, ``A guide to systematic review and
  meta-analysis of prediction model performance,'' {\em bmj}, vol.~356, 2017.

\bibitem{debray2019framework}
T.~P. Debray, J.~A. Damen, R.~D. Riley, K.~Snell, J.~B. Reitsma, L.~Hooft,
  G.~S. Collins, and K.~G. Moons, ``A framework for meta-analysis of prediction
  model studies with binary and time-to-event outcomes,'' {\em Statistical
  methods in medical research}, vol.~28, no.~9, pp.~2768--2786, 2019.

\bibitem{steingrimsson2021transporting}
J.~A. Steingrimsson, C.~Gatsonis, and I.~J. Dahabreh, ``Transporting a
  prediction model for use in a new target population,'' {\em arXiv preprint
  arXiv:2101.11182}, 2021.

\bibitem{dahabreh2020toward}
I.~J. Dahabreh, L.~C. Petito, S.~E. Robertson, M.~A. Hern{\'a}n, and J.~A.
  Steingrimsson, ``Toward causally interpretable meta-analysis: Transporting
  inferences from multiple randomized trials to a new target population,'' {\em
  Epidemiology}, vol.~31, no.~3, pp.~334--344, 2020.

\bibitem{dahabreh2023efficient}
I.~J. Dahabreh, S.~E. Robertson, L.~C. Petito, M.~A. Hern{\'a}n, and J.~A.
  Steingrimsson, ``Efficient and robust methods for causally interpretable
  meta-analysis: Transporting inferences from multiple randomized trials to a
  target population,'' {\em Biometrics}, vol.~79, no.~2, pp.~1057--1072, 2023.

\bibitem{shimodaira2000improving}
H.~Shimodaira, ``Improving predictive inference under covariate shift by
  weighting the log-likelihood function,'' {\em Journal of statistical planning
  and inference}, vol.~90, no.~2, pp.~227--244, 2000.

\bibitem{sugiyama2007covariate}
M.~Sugiyama, M.~Krauledat, and K.-R. M{\~A}{\v{z}}ller, ``Covariate shift
  adaptation by importance weighted cross validation,'' {\em Journal of Machine
  Learning Research}, vol.~8, no.~May, pp.~985--1005, 2007.

\bibitem{sugiyama2012machine}
M.~Sugiyama and M.~Kawanabe, {\em Machine learning in non-stationary
  environments: Introduction to covariate shift adaptation}.
\newblock MIT press, 2012.

\bibitem{morrison2023robust}
S.~Morrison, C.~Gatsonis, I.~J. Dahabreh, B.~Li, and J.~A. Steingrimsson,
  ``Robust estimation of loss-based measures of model performance under
  covariate shift,'' {\em Canadian Journal of Statistics}, 2023.

\bibitem{sahoo2022learning}
R.~Sahoo, L.~Lei, and S.~Wager, ``Learning from a biased sample,'' {\em arXiv
  preprint arXiv:2209.01754}, 2022.

\bibitem{angelopoulos2023prediction}
A.~N. Angelopoulos, S.~Bates, C.~Fannjiang, M.~I. Jordan, and T.~Zrnic,
  ``Prediction-powered inference,'' {\em Science}, vol.~382, no.~6671,
  pp.~669--674, 2023.

\bibitem{ge2023maximum}
J.~Ge, S.~Tang, J.~Fan, C.~Ma, and C.~Jin, ``Maximum likelihood estimation is
  all you need for well-specified covariate shift,'' {\em arXiv preprint
  arXiv:2311.15961}, 2023.

\bibitem{li2022estimating}
B.~Li, C.~Gatsonis, I.~J. Dahabreh, , and J.~A. Steingrimsson, ``Estimating the
  area under the roc curve when transporting a prediction model to a target
  population,'' {\em Biometrics}, 2022.

\bibitem{zhang2015multi}
K.~Zhang, M.~Gong, and B.~Sch{\"o}lkopf, ``Multi-source domain adaptation: A
  causal view,'' in {\em Twenty-ninth AAAI conference on artificial
  intelligence}, 2015.

\bibitem{sun2015survey}
S.~Sun, H.~Shi, and Y.~Wu, ``A survey of multi-source domain adaptation,'' {\em
  Information Fusion}, vol.~24, pp.~84--92, 2015.

\bibitem{zhao2020multi}
S.~Zhao, B.~Li, P.~Xu, and K.~Keutzer, ``Multi-source domain adaptation in the
  deep learning era: A systematic survey,'' {\em arXiv preprint
  arXiv:2002.12169}, 2020.

\bibitem{nomura2021efficient}
M.~Nomura and Y.~Saito, ``Efficient hyperparameter optimization under
  multi-source covariate shift,'' in {\em Proceedings of the 30th ACM
  International Conference on Information \& Knowledge Management},
  pp.~1376--1385, 2021.

\bibitem{mansour2008domain}
Y.~Mansour, M.~Mohri, and A.~Rostamizadeh, ``Domain adaptation with multiple
  sources,'' {\em Advances in neural information processing systems}, vol.~21,
  2008.

\bibitem{qiu2023efficient}
H.~Qiu, E.~T. Tchetgen, and E.~Dobriban, ``Efficient and multiply robust risk
  estimation under general forms of dataset shift,'' {\em arXiv preprint
  arXiv:2306.16406}, 2023.

\bibitem{brier1950verification}
G.~W. Brier, ``Verification of forecasts expressed in terms of probability,''
  {\em Monthly weather review}, vol.~78, no.~1, pp.~1--3, 1950.

\bibitem{robins1986new}
J.~Robins, ``A new approach to causal inference in mortality studies with a
  sustained exposure period—application to control of the healthy worker
  survivor effect,'' {\em Mathematical modelling}, vol.~7, no.~9-12,
  pp.~1393--1512, 1986.

\bibitem{youden1950index}
W.~J. Youden, ``Index for rating diagnostic tests,'' {\em Cancer}, vol.~3,
  no.~1, pp.~32--35, 1950.

\bibitem{scharfstein2018globalBiometrics}
D.~Scharfstein, A.~McDermott, I.~D{\'\i}az, M.~Carone, N.~Lunardon, and
  I.~Turkoz, ``Global sensitivity analysis for repeated measures studies with
  informative drop-out: A semi-parametric approach,'' {\em Biometrics},
  vol.~74, no.~1, pp.~207--219, 2018.

\bibitem{scharfstein2018globalSMMR}
D.~O. Scharfstein and A.~McDermott, ``Global sensitivity analysis of clinical
  trials with missing patient-reported outcomes,'' {\em Statistical Methods in
  Medical Research}, vol.~28, no.~5, pp.~1439--1456, 2018.

\bibitem{scharfstein2021global}
D.~O. Scharfstein, J.~Steingrimsson, A.~McDermott, C.~Wang, S.~Ray,
  A.~Campbell, E.~Nunes, and A.~Matthews, ``Global sensitivity analysis of
  randomized trials with nonmonotone missing binary outcomes: Application to
  studies of substance use disorders,'' {\em Biometrics}, 2021.

\bibitem{dahabreh2022global}
I.~J. Dahabreh, J.~M. Robins, S.~J. Haneuse, S.~E. Robertson, J.~A.
  Steingrimsson, and M.~A. Hern{\'a}n, ``Global sensitivity analysis for
  studies extending inferences from a randomized trial to a target
  population,'' {\em arXiv e-prints}, pp.~arXiv--2207, 2022.

\bibitem{luedtke2019omnibus}
A.~Luedtke, M.~Carone, and M.~J. van~der Laan, ``An omnibus non-parametric test
  of equality in distribution for unknown functions,'' {\em Journal of the
  Royal Statistical Society: Series B (Statistical Methodology)}, vol.~81,
  no.~1, pp.~75--99, 2019.

\bibitem{zeng2015random}
D.~Zeng and D.~Lin, ``On random-effects meta-analysis,'' {\em Biometrika},
  vol.~102, no.~2, pp.~281--294, 2015.

\bibitem{steingrimsson2023systematically}
J.~A. Steingrimsson, D.~H. Barker, R.~Bie, and I.~J. Dahabreh, ``Systematically
  missing data in causally interpretable meta-analysis,'' {\em Biostatistics},
  2023.

\bibitem{robins2000sensitivity}
J.~M. Robins, A.~Rotnitzky, and D.~O. Scharfstein, ``Sensitivity analysis for
  selection bias and unmeasured confounding in missing data and causal
  inference models,'' in {\em Statistical models in epidemiology, the
  environment, and clinical trials}, pp.~1--94, Springer, 2000.

\bibitem{national2011reduced}
N.~L. S. T.~R. Team, ``Reduced lung-cancer mortality with low-dose computed
  tomographic screening,'' {\em New England Journal of Medicine}, vol.~365,
  no.~5, pp.~395--409, 2011.

\bibitem{oken2011screening}
M.~M. Oken, W.~G. Hocking, P.~A. Kvale, G.~L. Andriole, S.~S. Buys, T.~R.
  Church, E.~D. Crawford, M.~N. Fouad, C.~Isaacs, D.~J. Reding, {\em et~al.},
  ``Screening by chest radiograph and lung cancer mortality: the prostate,
  lung, colorectal, and ovarian (plco) randomized trial,'' {\em JAMA},
  vol.~306, no.~17, pp.~1865--1873, 2011.

\bibitem{krist2021screening}
A.~H. Krist, K.~W. Davidson, C.~M. Mangione, M.~J. Barry, M.~Cabana, A.~B.
  Caughey, {\em et~al.}, ``Screening for lung cancer: Us preventive services
  task force recommendation statement,'' {\em JAMA}, vol.~325, no.~10,
  pp.~962--970, 2021.

\bibitem{shao2003impact}
J.~Shao, ``Impact of the bootstrap on sample surveys,'' {\em Statistical
  Science}, vol.~18, no.~2, pp.~191--198, 2003.

\bibitem{steingrimsson2022systematically}
J.~A. Steingrimsson, D.~H. Barker, R.~Bie, and I.~J. Dahabreh, ``Systematically
  missing data in causally interpretable meta-analysis,'' {\em arXiv preprint
  arXiv:2205.00610}, 2022.

\end{thebibliography}
%%%%%%%%%%%%%%%%%%%%%%%%%%%%%%%%%%%%%%%%%%%%%%%%%%%%%%%%%%%%%%%%%%%%%%%%%%%%%%

%%%%%%%%%%%%%%%%%%%%%%%%%%%%%%%%%%%%%%%%%%%%%%%%%%%%%%%%%%%%
\newpage 
\appendix 
%%%%%%%%%%%%%%%%%%%%%%%%%%%%%%%%%%%%%%%%%%%%%%%%%%%%%%%%%%%%

\renewcommand{\theequation}{A.\arabic{equation}}
\setcounter{equation}{0}

\section{Proofs}
\label{app:proofs}
\textbf{Proof of theorem \ref{thm:id}:} For the outcome model representation
\begin{align*}
\E[I(h(X, \widehat \beta) > c)|Y = 1, S=0] &= \frac{\E[I(h(X, \widehat \beta) > c, Y = 1)|S=0]}{\E[I(Y=1)|S=0]} \\
&= \frac{\E[\E[I(h(X, \widehat \beta) > c, Y = 1)|X, S=0]|S=0]}{\E[\E[I(Y=1)|X,S=0]|S=0]} \\
&= \frac{\E[\E[I(h(X, \widehat \beta) > c, Y = 1)|X, S=0]|S=0]}{\E[\E[I(Y=1)|X,R=1]|S=0]} \\
&= \frac{\E[I(h(X, \widehat \beta) > c)\E[I(Y = 1)|X, S=0]|S=0]}{\E[\E[I(Y=1)|X,R=1]|S=0]} \\
&= \frac{\E[I(h(X, \widehat \beta) > c) \Pr[Y = 1|X, R=1]|S=0]}{\E[\Pr[Y=1|X,R=1]|S=0]}.
\end{align*}
For the inverse-odds weighting we have\allowdisplaybreaks
\begin{align*}
\E[I(h(X, \widehat \beta) > c)|Y = 1, S=0] &= \frac{\E[I(h(X, \widehat \beta) > c, Y = 1)|S=0]}{\E[I(Y=1)|S=0]} \\
&= \frac{\E[\E[I(h(X, \widehat \beta) > c, Y = 1)|X,S=0]|S=0]}{\E[\E[I(Y=1)|X,S=0]|S=0]} \\
&= \frac{\E[I(S=0)I(h(X, \widehat \beta) > c) \Pr[Y = 1|X,S=0]]}{\E[I(S=0)\Pr[Y=1|X,S=0]]} \\
&= \frac{\E[I(S=0)I(h(X, \widehat \beta) > c) \Pr[Y = 1|X,R=1]]}{\E[I(S=0)\Pr[Y=1|X,R=1]]} \\
&= \frac{\E\left[I(S=0)I(h(X, \widehat \beta) > c) \frac{\Pr[Y = 1,R=1|X]}{\Pr[R=1|X]}\right]}{\E\left[I(S=0)\frac{\Pr[Y=1,R=1|X]}{\Pr[R=1|X]}\right]} \\
&= \frac{\E\left[\E\left[I(S=0)I(h(X, \widehat \beta) > c) \frac{\Pr[Y = 1,R=1|X]}{\Pr[R=1|X]}\Big|X\right]\right]}{\E\left[\E\left[I(S=0)\frac{\Pr[Y=1,R=1|X]}{\Pr[R=1|X]}\Big|X \right]\right]} \\
&= \frac{\E\left[\E\left[\frac{\Pr[R=0|X]}{\Pr[R=1|X]}I(h(X, \widehat \beta) > c) \Pr[Y = 1,R=1|X]\Big|X\right]\right]}{\E\left[\E\left[\frac{\Pr[S=0|X]}{\Pr[R=1|X]}\Pr[Y=1,R=1|X]\Big|X \right]\right]} \\
&= \frac{\E\left[\E\left[\frac{\Pr[R=0|X]}{\Pr[R=1|X]}I(h(X, \widehat \beta) > c, Y=1, R=1) \Big|X\right]\right]}{\E\left[\E\left[\frac{\Pr[S=0|X]}{\Pr[R=1|X]}I(Y=1,R=1)\Big|X \right]\right]} \\
&= \frac{\E\left[\frac{\Pr[S=0|X]}{\Pr[R=1|X]}I(h(X, \widehat \beta) > c, Y=1, R=1) \right]}{\E\left[\frac{\Pr[S=0|X]}{\Pr[R=1|X]}I(Y=1,R=1)\right]}
\end{align*}

\noindent \textbf{Proof of theorem \ref{thm-id-auc}:} Let $l$ denote a random observation from the target population with the outcome ($Y=1$) and $k$ denote a random observation from the target population without the outcome ($Y=0$). Following \cite{li2022estimating} we have that the AUC in the target population can be rewritten as
\allowdisplaybreaks
\begin{align*}
&\normalfont\E[I(h( X _l, \widehat \beta) > h( X _k, \widehat \beta))|Y _l=1, Y _k=0, S _l=0, S _k=0] \\ 
=& \frac{\normalfont\E[I(h( X _l, \widehat \beta) > h( X _k, \widehat \beta), Y _l=1, Y _k=0)|S _l=0, S _k=0]}{\Pr[Y _l=1, Y _k=0| S _l=0, S _k=0]} \\ 
=& \frac{ \normalfont\E[\normalfont\E[I(h( X _l, \widehat \beta) > h( X _k, \widehat \beta), Y _l=1, Y _k=0)| X _l,  X _k, S _l=0, S _k=0] |S _l=0, S _k=0]}{\normalfont\E[\Pr[Y _l=1, Y _k=0|X _l,  X _k, S _l=0, S _k=0]|S _l=0, S _k=0]} \\
= & \frac{ \normalfont\E[I(S _l=0, S _k=0, h( X _l, \widehat \beta) > h( X _k, \widehat \beta))\Pr[Y _l=1, Y _k=0| X _l,  X _k, S _l=0, S _k=0]]}{\normalfont\E[I(S _l=0, S _k=0)\Pr[Y _l=1, Y _k=0|X _l,  X _k, S _l=0, S _k=0]]} \\ 
= & \frac{ \normalfont\E[I(S _l=0, S _k=0, h( X _l, \widehat \beta) > h( X _k, \widehat \beta))\Pr[Y _l=1, Y _k=0| X _l,  X _k, R _l=1, R _k=1]]}{\normalfont\E[I(S _l=0, S _k=0)\Pr[Y _l=1, Y _k=0|X _l,  X _k, R _l=1, R _k=1]]} \\
= & \frac{ \normalfont\E[I(S _l=0, S _k=0, h( X _l, \widehat \beta) > h( X _k, \widehat \beta))\Pr[Y _l=1| X _l, R _l=1]\Pr[Y _k=0| X _k, R _k=1]]}{\normalfont\E[I(S _l=0, S _k=0)\Pr[Y _l=1|X _l, R _l=1]\Pr[Y _k=0|R _k=1,  X _k]]} \\
= & \frac{ \normalfont\E[I(h( X _l, \widehat \beta) > h( X _k, \widehat \beta))\Pr[Y _l=1| X _l, R _l=1]\Pr[Y _k=0| X _k, R _k=1]|S _l=0, S _k=0]}{\normalfont\E[\Pr[Y _l=1| X _l,R _l=1]\Pr[Y _k=0|R _k=1,  X _k]|S _l=0, S _k=0]}.
\end{align*}
Now we prove the weighting-based identifiability result given by expression \eqref{AUC-ID-W} \allowdisplaybreaks
\begin{align*}
&\normalfont\E[I(h( X _l, \widehat \beta) > h( X _k, \widehat \beta))|Y _l=1, Y _k=0, S _l=0, S _k=0] \\  
=& \frac{\normalfont\E[I(h( X _l, \widehat \beta) > h( X _k, \widehat \beta), Y _l=1, Y _k=0)|S _l=0, S _k=0]}{\Pr[Y _l=1, Y _k=0| S _l=0, S _k=0]}\\ 
=& \frac{\normalfont\E[\normalfont\E[I(h( X _l, \widehat \beta) > h( X _k, \widehat \beta), Y _l=1, Y _k=0)| X _l,  X _k, S _l=0, S _k=0] |S _l=0, S _k=0]}{\normalfont\E[\Pr[Y _l=1, Y _k=0|X _l,  X _k, S _k=0, S _k=0]|S _l=0, S _k=0]} \\ 
= & \frac{\normalfont\E[I(S _l=0, S _k=0)I(h( X _l, \widehat \beta) > h( X _k, \widehat \beta)) \Pr[Y _l=1, Y _k=0| X _l,  X _k, S _l=0, S _k=0]]}{\normalfont\E[I(S _l=0, S _k=0)\Pr[Y _l=1, Y _k=0|X _l,  X _k, S _k=0, S _k=0]]}  \\
= & \frac{\normalfont\E[I(S _l=0, S _k=0)I(h( X _l, \widehat \beta) > h( X _k, \widehat \beta)) \Pr[Y _l=1, Y _k=0| X _l,  X _k, R _l=1, R _k=1]]}{\normalfont\E[I(S _l=0, S _k=0)\Pr[Y _l=1, Y _k=0|X _l,  X _k, R _l=1, R _k=1]]}  \\
= & \frac{\normalfont\E\left[I(S _l=0, S _k=0)I(h( X _l, \widehat \beta) > h( X _k, \widehat \beta))\frac{\Pr[Y _l=1, Y _k=0, R _l=1, R _k=1| X _l,  X _k]}{\Pr[R _l=1, R _k=1| X _l,  X _k]}\right]}{\normalfont\E\left[I(S _l=0, S _k=0)\frac{\Pr[Y _l=1, Y _k=0, R _l=1, R _k=1| X _l,  X _k]]}{\Pr[R _l=1, R _k=1| X _l,  X _k]}\right]}  \\
= & \frac{\normalfont\E\left[\normalfont\E\left[I(S _l=0, S _k=0)I(h( X _l, \widehat \beta) > h( X _k, \widehat \beta))\frac{\Pr[Y _l=1, Y _k=0, R _l=1, R _k=1| X _l,  X _k]}{\Pr[R _l=1, R _k=1| X _l,  X _k]}\big| X _l, X _k\right]\right]}{\normalfont\E\left[\normalfont\E\left[I(S _l=0, S _k=0)\frac{\Pr[Y _l=1, Y _k=0, R _l=1, R _k=1| X _l,  X _k]]}{\Pr[R _l=1, R _k=1| X _l,  X _k]}\big| X _l, X _k\right]\right]}  \\
= &\frac{\normalfont\E\left[\frac{\Pr[S=0| X _l]\Pr[S=0| X _k]}{\Pr[R=1| X _l]\Pr[R=1| X _k]} I(h( X _l, \widehat \beta) > h( X _k, \widehat \beta)) \Pr[Y _l=1, Y _k=0, R _l=1, R _k=1| X _l,  X _k]\right]}{\normalfont\E\left[\frac{\Pr[S=0| X _l]\Pr[S=0| X _k]}{\Pr[R=1| X _l]\Pr[R=1| X _k]}\Pr[Y _l=1, Y _k=0, R _l=1, R _k=1| X _l,  X _k]\right]} \\
= &\frac{\normalfont\E\left[\frac{\Pr[S=0| X _l]\Pr[S=0| X _k]}{\Pr[R=1| X _l]\Pr[R=1| X _k]} I(h( X _l, \widehat \beta) > h( X _k, \widehat \beta), Y _l=1, Y _k=0, R _l=1, R _k=1)\right]}{\normalfont\E\left[\frac{\Pr[S=0| X _l]\Pr[S=0| X _k]}{\Pr[R=1| X _l]\Pr[R=1| X _k]}I(Y _l=1, Y _k=0, R _l=1, R _k=1)\right]}
\end{align*}

\section{Influence function for sensitivity}
\label{app:send-inf}
We start by finding the non-parametric influence function for the denominator and numerator in the outcome formulation for sensitivity 
\[
\frac{\E[\E[I(h(X, \widehat \beta) > c) I(Y = 1)|X, R=1]|S=0]}{\E[\E[Y=1|X,R=1]|S=0]} 
=: \frac{\alpha_1}{\alpha_0}
\]
Following the same steps as in the proof of theorem 3 in \cite{morrison2023robust} we get the non-parametric influence function for the numerator is
\begin{align*}
\alpha_{1, p_0}^1 &= \frac{I(S=0)}{\Pr[S=0]} \big\{E[I(h(X, \widehat \beta) > c, Y = 1)|X, R=1] - \alpha_1 \big\}  \\
 & +  \frac{\Pr[S=0|X] I(R=1)}{\Pr[S=0] \Pr[R=1|X]} \big\{ I(h(X, \widehat \beta) > c, Y = 1) - \E[I(h(X, \widehat \beta) > c) I(Y = 1)|X, R=1] \big\} \\
 &= \frac{I(S=0)}{\Pr[S=0]} \big\{I(h(X, \widehat \beta) > c) E[I(Y = 1)|X, R=1] - \alpha_1 \big\}  \\
 & +  I(h(X, \widehat \beta) > c)\frac{\Pr[S=0|X] I(R=1)}{\Pr[S=0] \Pr[R=1|X]} \big\{ I(Y = 1) - \E[I(Y = 1)|X, R=1] \big\} 
\end{align*}
and the non-parametric influence function for the denominator is 
\begin{align*}
\alpha_{0,p_0}^1 &= \frac{I(S=0)}{\Pr[S=0]} \big\{E[I(Y = 1)|X, R=1] - \alpha_0 \big\}  \\
 & +  \frac{\Pr[S=0|X] I(R=1)}{\Pr[S=0] \Pr[R=1|X]} \big\{ I(Y = 1) - \E[I(Y = 1)|X, R=1] \big\}
\end{align*}
The influence function $\alpha_{1, p_0}^1$ leads to the non-parametric estimator for the numerator
\begin{align*}
\frac{1}{n_0} \sum_{i=1}^n \left(I(S_i=0) I(h(X_i, \widehat \beta) > c) \widehat m(X_i) +  \widehat w(X_i) I(R_i=1) I(h(X_i, \widehat \beta) > c) \big\{ I(Y_i = 1) - \widehat m(X_i)\big\}\right),.
\end{align*}
And, The influence function $\alpha_{0, p_0}^1$ leads to the non-parametric estimator for the denominator
\begin{align*}
\frac{1}{n_0} \sum_{i=1}^n \left(I(S_i=0) \widehat m(X_i) +  \widehat w(X_i) I(R_i=1) \big\{ I(Y_i = 1) - \widehat m(X_i)\big\}\right)
\end{align*}
Hence, a natural estimator for the sensitivity in the target population is 
\[
\widehat \psi_{sens, dr} = \frac{\sum_{i=1}^n \left(I(S_i=0) I(h(X_i, \widehat \beta) > c) \widehat m(X_i) +  \widehat w(X_i) I(R_i=1) I(h(X_i, \widehat \beta) > c) \big\{ I(Y_i = 1) - \widehat m(X_i)\big\}\right)}{\sum_{i=1}^n \left(I(S_i=0) \widehat m(X_i) +  \widehat w(X_i) I(R_i=1) \big\{ I(Y_i = 1) - \widehat m(X_i)\big\}\right)}
\]
The non-parametric influence function for the sensitivity in the target population is 
\begin{align*}
\frac{\alpha_{1, p_0}^1}{\alpha_0} - \frac{\alpha_1 \alpha_{0,p_0}^1}{\alpha_0^2} &=  \frac{1}{\alpha_0} \left(\alpha_{1, p_0}^1 - \psi_{sens} \alpha_{0, p_0}^1 \right) 
\end{align*}
Now 
\begin{align*}
- \frac{I(S=0)}{\Pr[S=0]} \alpha_1 + \psi_{sens} \frac{I(S=0)}{\Pr[S=0]} \alpha_0 = 0
\end{align*}
So 
\begin{align*}
\alpha_{1, p_0}^1 - \psi_{sens} \alpha_{0, p_0}^1 &= A - B,
\end{align*}
where
\begin{align*}
A &= \frac{I(S=0)}{\Pr[S=0]} \big\{I(h(X, \widehat \beta) > c) E[I(Y = 1)|X, R=1] \big\}  \\
 & +  I(h(X, \widehat \beta) > c)\frac{\Pr[S=0|X] I(R=1)}{\Pr[S=0] \Pr[R=1|X]} \big\{ I(Y = 1) - \E[I(Y = 1)|X, R=1] \big\} \\  
 &= I(h(X, \widehat \beta) > c) \bigg(\frac{I(S=0)}{\Pr[S=0]}  E[I(Y = 1)|X, R=1]   \\
 & +  \frac{\Pr[S=0|X] I(R=1)}{\Pr[S=0] \Pr[R=1|X]} \big\{ I(Y = 1) - \E[I(Y = 1)|X, R=1] \big\} \bigg)
\end{align*}
\begin{align*}
B = \psi_{sens} \left(\frac{I(S=0)}{\Pr[S=0]} E[I(Y = 1)|X, R=1] +  \frac{\Pr[S=0|X] I(R=1)}{\Pr[S=0] \Pr[R=1|X]} \big\{ I(Y = 1) - \E[I(Y = 1)|X, R=1] \big\}\right)
\end{align*}
Hence, 
\begin{align*}
&\alpha_{1, p_0}^1 - \psi_{sens} \alpha_{0, p_0}^1 = (I(h(X, \widehat \beta) > c) - \psi_{sens}) \times 
\\ &\left(\frac{I(S=0)}{\Pr[S=0]} E[I(Y = 1)|X, R=1] +  \frac{\Pr[S=0|X] I(R=1)}{\Pr[S=0] \Pr[R=1|X]} \big\{ I(Y = 1) - \E[I(Y = 1)|X, R=1] \big\}\right).
\end{align*}
And the non-parametric influence function for sensitivity is
\[
\frac{(I(h(X, \widehat \beta) > c) - \psi_{sens}) \left(\frac{I(S=0)}{\Pr[S=0]} E[I(Y = 1)|X, R=1] +  \frac{\Pr[S=0|X] I(R=1)}{\Pr[S=0] \Pr[R=1|X]} \big\{ I(Y = 1) - \E[I(Y = 1)|X, R=1] \big\}\right)}{\E[\E[Y=1|X,R=1]|S=0]}
\]

\section{Consistency}
\label{app:asym-prop}
Now we prove Theorem \ref{thm-sens-asymp} for the consistency of the doubly robust estimator for sensitivity in the target population.
\begin{proof}
The doubly robust estimator for sensitivity is
\[
\widehat \psi_{sens, dr} = \frac{\sum_{i=1}^n \left(I(S_i=0) I(h(X_i, \widehat \beta) > c) \widehat m(X_i) +  \widehat w(X_i) I(R_i=1) I(h(X_i, \widehat \beta) > c) \big\{ I(Y_i = 1) - \widehat m(X_i)\big\}\right)}{\sum_{i=1}^n \left(I(S_i=0) \widehat m(X_i) + \widehat w(X_i) I(R_i=1)\big\{ I(Y_i = 1) - \widehat m(X_i)\big\}\right)}
\]
Let $m^*(X)$ be the asymptotic limit of $\widehat m(X)$ and $w^*(X)$ be the asymptotic limit of $\widehat w(X)$. As $n \longrightarrow \infty$ the doubly robust estimator for sensitivity $\widehat \psi_{sens, dr}$ converges by the continuous mapping theorem to
\[
\frac{\Pr[S=0]^{-1}\E\left[I(S=0) I(h(X, \widehat \beta) > c) m^*(X) +  w^*(X) I(R=1) I(h(X, \widehat \beta) > c) \big\{ I(Y = 1) - m^*(X)\big\}\right]}{\Pr[S=0]^{-1}\E \left[I(S=0) m^*(X) +  w^*(X) I(R=1) \big\{ I(Y = 1) - m^*(X)\big\}\right]}.
\]
First assume that $m^*(X) = \Pr[Y = 1|X, R=1]$ but we do not assume that $p^*(X) = \Pr[R=1|X]$. Under that assumption 
\begin{align*}
\Pr[S=0]^{-1}\E\left[I(S=0) I(h(X, \widehat \beta) > c) m^*(X)\right] &= \Pr[S=0]^{-1} \E\left[I(S=0) I(h(X, \widehat \beta) > c) \Pr[Y = 1|X, R=1]\right] \\
&= \E\left[I(h(X, \widehat \beta) > c) \Pr[Y = 1|X, R=1]\big|S=0\right]
\end{align*}
and 
\begin{align*}
&E\left[w^*(X) I(R=1) I(h(X, \widehat \beta) > c) \big\{ I(Y = 1) - m^*(X)\big\}\right] \\ &= E\left[\E\left[w^*(X) I(R=1) I(h(X, \widehat \beta) > c) \big\{ I(Y = 1) - \Pr[Y = 1|X, R=1] \big\} \big|X \right]\right] \\
&= E\left[w^*(X) I(h(X, \widehat \beta) > c)\E\left[ I(R=1) \big\{ I(Y = 1) - \Pr[Y = 1|X, R=1] \big\} \big|X \right]\right] \\
&= E\left[w^*(X)  I(h(X, \widehat \beta) > c)\Pr[R=1|X] \E\left[ \big\{ I(Y = 1) - \Pr[Y = 1|X, R=1] \big\} \big|X, R=1 \right]\right] \\ &= 0
\end{align*}
Combining this gives 
\begin{align*}
&\Pr[S=0]^{-1}\E\left[I(S=0) I(h(X, \widehat \beta) > c) m^*(X) +  w^*(X) I(R=1) I(h(X, \widehat \beta) > c) \big\{ I(Y = 1) - m^*(X)\big\}\right] \\ &= \E\left[I(h(X, \widehat \beta) > c) \Pr[Y = 1|X, R=1]\big|S=0\right].
\end{align*}
The same arguments give 
\begin{align*}
&\Pr[S=0]^{-1}\E \left[I(S=0) m^*(X) +  w^*(X) I(R=1) \big\{ I(Y = 1) - m^*(X)\big\}\right] \\
&= \E\left[\Pr[Y = 1|X, R=1]\big|S=0\right].
\end{align*}
And it follows from the identifiability results of the sensitivity in the target population that
\[
\widehat \psi_{sens, dr} \longrightarrow \E[I(h(X, \widehat \beta) > c)|Y = 1, S=0]
\]
in probability.
\\
\\
Now assume that $p^*(X) = \Pr[R=1|X]$ but we do not assume that $m^*(X) = \Pr[Y = 1|X, R=1]$. Under that assumption
\begin{align*}
&\Pr[S=0]^{-1}\E\left[\frac{(1-\Pr[R=1|X]) I(R=1) I(h(X, \widehat \beta) > c)}{ \Pr[R=1|X]} I(Y = 1)\big\}\right] \\ &= \E[I(h(X, \widehat \beta) > c, Y = 1)|S=0]
\end{align*}
by the results in the proof of Theorem \ref{thm:id}. Now 
\begin{align*}
&\E\left[ I(S=0) I(h(X, \widehat \beta) > c) m^*(X) - \frac{(1-\Pr[R=1|X]) I(R=1) I(h(X, \widehat \beta) > c)}{ \Pr[R=1|X]} m^*(X)\big\}\right] \\ 
&= \E\left[\E\left[ I(S=0) I(h(X, \widehat \beta) > c) m^*(X) - \frac{(1-\Pr[R=1|X]) I(R=1) I(h(X, \widehat \beta) > c)}{ \Pr[R=1|X]} m^*(X)\big\}\bigg| X \right]\right] \\
&= \E\left[ E[I(S=0)|X] I(h(X, \widehat \beta) > c) m^*(X) - \frac{(1-\Pr[R=1|X]) I(h(X, \widehat \beta) > c)}{ \Pr[R=1|X]} m^*(X) E[I(R=1)|X]\big\}\right] \\
&= 0.
\end{align*}
Hence,
\begin{align*}
&Pr[S=0]^{-1}\E\left[I(S=0) I(h(X, \widehat \beta) > c) m^*(X) +  w^*(X) I(R=1) I(h(X, \widehat \beta) > c) \big\{ I(Y = 1) - m^*(X)\big\}\right]\\ &= \E[I(h(X, \widehat \beta) > c, Y = 1)|S=0]
\end{align*}
and by similar arguments 
\begin{align*}
&Pr[S=0]^{-1}\E\left[I(S=0) m^*(X) +  \frac{\Pr[S=0|X] I(R=1)}{ \Pr[R=1|X]} \big\{ I(Y = 1) - m^*(X)\big\}\right] \\ &= \E[I(Y = 1)|S=0].
\end{align*}
And it follows from the identifiability results of the sensitivity in the target population that
\[
\widehat \psi_{sens, dr} \longrightarrow \E[I(h(X, \widehat \beta) > c)|Y = 1, S=0]
\]
in probability.
\end{proof}
Now we prove Theorem \ref{thm-auc-asymp} for the consistency of the doubly robust estimator for AUC in the target population.
\begin{proof}
The result follows from the continuous mapping theorem if we can show that
$$\frac{1}{n(n-1)} \sum_{i \neq j} d^{dr}(O_i, O_j,  I(h(X_i, \widehat \beta) > h(X_j, \widehat \beta))$$ and the denominator $$\frac{1}{n(n-1)} \sum_{i \neq j} d^{dr}(O_i, O_j, 1)$$ are consistent

First, suppose that $w^*(X) = \frac{\Pr[R=0|X]}{\Pr[R=1|X]}$ but we dont make any assumptions on the asymptotic limit $m^*(X)$. Minor adaptations of  Theorem 1 in  \cite{li2022estimating} show that $\widehat \psi_{auc, w}$ is consistent when $w^*(X) = \frac{\Pr[R=0|X]}{\Pr[R=1|X]}$ so it is enough to show that for $k(X_i,X_j) = I(h(X_i, \widehat \beta) > h(X_j, \widehat \beta)$ and $k(X_i,X_j) =1$ then 
\begin{equation*} 
\frac{1}{n(n-1)}\sum_{i \neq j}  \widehat m(X_i) (1 - \widehat m(X_j)) (I(S_i=0, S_j=0)- I(R_i=1, R_j=1) \widehat w(X_i) \widehat w(X_j)) k(X_i,X_j) \overset{P}{\longrightarrow} 0.
\end{equation*} 
Now as $ k(X_i,X_j)$ only depends on $(X_i, X_j)$
\begin{align*}
&\E\left[  m^*(X_i) (1 - m^*(X_j)) I(R_i=1, R_j=1) w^*(X_i) w^*(X_j) k(X_i,X_j)  \right] \\
&= \E\left[ \E\left[ m^*(X_i) (1-m^*(X_j)) I(R_i=1, R_j=1) w^*(X_i) w^*(X_j) k(X_i,X_j) \big| X_i, X_i \right] \right] \\
&= \E\left[  m^*(X_i) (1-m^*(X_j)) \E[I(R_i=1, R_j=1)|X_i, X_j] \frac{\Pr[R=0|X_i]}{\Pr[R=1|X_i]} \frac{\Pr[R=0|X_j]}{\Pr[R=1|X_j]} k(X_i,X_j) \right]  \\
&= \E\left[  m^*(X_i) (1-m^*(X_j)) \Pr[R=0|X_i]\Pr[R=0|X_j] k(X_i,X_j) \right] 
\end{align*}
and
\begin{align*}
&\E[ m^*(X_i) (1-m^*(X_j)) I(S_i=0, S_j=0) k(X_i,X_j)] \\
&= \E[ \E[ m^*(X_i) (1-m^*(X_j)) I(S_i=0, S_j=0) k(X_i,X_j)| X_i, X_j]] \\
&= \E[  m^*(X_i) (1-m^*(X_j)) \E[I(S_i=0, S_j=0)|X_i, X_j] k(X_i,X_j)] \\
&= \E[  m^*(X_i) (1-m^*(X_j)) \Pr[R=0|X_i] \Pr[R=0|X_j] k(X_i,X_j)].
\end{align*}
It follows that 
\[
\E[m^*(X_i) (1-m^*(X_j)) (I(S_i=0, S_j=0)- I(R_i=1, R_j=1) w^*(X_i) w^*(X_j)) k(X_i,X_j) ] = 0.
\]
By Lemma 2 in \cite{li2022estimating} $\widehat m(X_i) (1 - \widehat m(X_j))$ and $\widehat m(X_i) (1- \widehat m(X_j)) \widehat w(X_i) \widehat w(X_j)$ are stochastically equicontinuous and the results follows from applying Lemma 1 in \cite{li2022estimating}.
\\
\\
Now, assume that $m^*(X) = \Pr[Y=1|X,m R=1]$ but we don't assume that $w^*(X) = \frac{\Pr[R=0|X]}{\Pr[R=1|X]}$. Minor adaptations of  Theorem 2 in  \cite{li2022estimating} show that $\widehat \psi_{auc, out}$ is consistent when $m^*(X) = \Pr[Y=1|X,m R=1]$  so it is enough to show that for $k(X_i,X_j) = I(h(X_i, \widehat \beta) > h(X_j, \widehat \beta)$ and $k(X_i,X_j) =1$ then 
\begin{equation*} 
\frac{1}{n(n-1)}\sum_{i \neq j}  \widehat w(X_i) \widehat w(X_j) (I(Y_i=1, Y_j=0) - \widehat m(X_i) (1 - \widehat m(X_j))) I(R_i=1, R_j=1) k(X_i,X_j) \overset{P}{\longrightarrow} 0.
\end{equation*}
Now 
\begin{align*}
&\E[w^*(X_i) w^*(X_j) (I(Y_i=1, Y_j=0)I(R_i=1, R_j=1) - m^*(X_i) (1 - m^*(X_j)))   k(X_i,X_j)] \\
&= \E[\E[w^*(X_i) w^*(X_j) (I(Y_i=1, Y_j=0) - m^*(X_i) (1- m^*(X_j)))  I(R_i=1, R_j=1) k(X_i,X_j)|X_i, X_i]] \\
&= \E[w^*(X_i) w^*(X_j) (\E[I(Y_i=1, Y_j=0)|X_i, X_i, R_i=1, R_j=1] \\
&- \Pr[Y_i=1|X_i, R_i=1] \Pr[Y_i=j|X_j, R_j=1] )  \Pr[R_i=1|X_i] \Pr[R_j=1|X_j] k(X_i,X_j)] \\
&= 0
\end{align*}
The results follows from applying Lemma 1 in \cite{li2022estimating}.
\end{proof}
 
\section{Results for specificity, negative and positive predictive value, and estimators of risk}
\label{app:NPV}

\subsection{Identification}

The following theorem shows identifiability of specificity in the target population is given below and proved in Supplementary Web Appendix \ref{app:proofs}.
\begin{theorem}
\label{thm:id-spec}
If assumptions A1 and A2 and that $\E[\Pr[Y=0|X,R=1]|S=0]>0$ hold, then the specificity in the target population is identifiable using the observed data through the observed data functional 
\begin{equation}
\label{id:out-spec}
\frac{\E[I(h(X, \widehat \beta) \leq c) \Pr[Y = 0|X, R=1]|R=0]}{\E[\Pr[Y=0|X,R=1]R=0]}
\end{equation}
or equivalently using the weighting representation
\begin{equation}
\label{id:ipw-spec}
\frac{\E\left[\frac{\Pr[R=0|X]}{\Pr[R=1|X]}I(h(X, \widehat \beta) \leq c, Y=0, R=1) \right]}{\E\left[\frac{\Pr[R=0|X]}{\Pr[R=1|X]}I(Y=0,R=1)\right]}.
\end{equation}
\end{theorem}
\begin{proof}
For the outcome model representation
\begin{align*}
\E[I(h(X, \widehat \beta) \leq c)|Y = 0, S=0] &= \frac{\E[I(h(X, \widehat \beta) \leq c, Y = 0)|S=0]}{\E[I(Y=0)|S=0]} \\
&= \frac{\E[\E[I(h(X, \widehat \beta) \leq c, Y = 0)|X, S=0]|S=0]}{\E[\E[I(Y=1)|X,S=0]|S=0]} \\
&= \frac{\E[\E[I(h(X, \widehat \beta) \leq c, Y = 0)|X, R=1]|S=0]}{\E[\E[I(Y=1)|X,R=1]|S=0]} \\
&= \frac{\E[I(h(X, \widehat \beta) \leq c) \Pr[Y = 0|X, R=1]|S=0]}{\E[\Pr[Y=1|X,R=1]|S=0]}.
\end{align*}
For the inverse-odds weighting we have
\begin{align*}
\E[I(h(X, \widehat \beta) \leq c)|Y = 0, S=0] &= \frac{\E[I(h(X, \widehat \beta) \leq c, Y = 0)|S=0]}{\E[I(Y = 0)|S=0]} \\
&= \frac{\E[\E[I(h(X, \widehat \beta) \leq c, Y = 0)|X,S=0]|S=0]}{\E[\E[I(Y = 0)|X,S=0]|S=0]} \\
&= \frac{\E[I(S=0)I(h(X, \widehat \beta) \leq c) \Pr[Y = 0|X,S=0]]}{\E[I(S=0)\Pr[Y = 0|X,S=0]]} \\
&= \frac{\E[I(S=0)I(h(X, \widehat \beta) \leq c) \Pr[Y = 0|X,R=1]]}{\E[I(S=0)\Pr[Y = 0|X,R=1]]} \\
&= \frac{\E\left[I(S=0)I(h(X, \widehat \beta) \leq c) \frac{\Pr[Y = 0,R=1|X]}{\Pr[R=1|X]}\right]}{\E\left[I(S=0)\frac{\Pr[Y = 0,R=1|X]}{\Pr[R=1|X]}\right]} \\
&= \frac{\E\left[\E\left[I(S=0)I(h(X, \widehat \beta) \leq c) \frac{\Pr[Y = 0,R=1|X]}{\Pr[R=1|X]}\Big|X\right]\right]}{\E\left[\E\left[I(S=0)\frac{\Pr[Y = 0,R=1|X]}{\Pr[R=1|X]}\Big|X \right]\right]} \\
&= \frac{\E\left[\E\left[\frac{\Pr[S=0|X]}{\Pr[R=1|X]}I(h(X, \widehat \beta) \leq c) \Pr[Y = 0,R=1|X]\Big|X\right]\right]}{\E\left[\E\left[\frac{\Pr[S=0|X]}{\Pr[R=1|X]}\Pr[Y = 0,R=1|X]\Big|X \right]\right]} \\
&= \frac{\E\left[\E\left[\frac{\Pr[S=0|X]}{\Pr[R=1|X]}I(h(X, \widehat \beta) \leq c, Y = 0, R=1) \Big|X\right]\right]}{\E\left[\E\left[\frac{\Pr[S=0|X]}{\Pr[R=1|X]}I(Y = 0,R=1)\Big|X \right]\right]} \\
&= \frac{\E\left[\frac{\Pr[S=0|X]}{\Pr[R=1|X]}I(h(X, \widehat \beta) \leq c, Y = 0, R=1) \right]}{\E\left[\frac{\Pr[S=0|X]}{\Pr[R=1|X]}I(Y = 0,R=1)\right]}
\end{align*}
\end{proof}
\noindent Recall that the positive predicted value in the target population is defined as $\Pr[Y=1|I(h(X, \widehat \beta) > c), S=0]$ and the negative predicted value in the target population is defined as $\Pr[Y=0|I(h(X, \widehat \beta) \leq c), S=0]$. The following two theorems show that the positive predictive value and negative predictive value are identifiable.
\begin{theorem}
\label{thm:id-ppv}
Under assumptions A1 and A2 and that $\E[I(h(X, \widehat \beta) > c)|S=0]>0$ the positive predictive value in the target population is identifiable using the observed data through the observed data functional 
\begin{equation}
\label{id:out-ppv}
\frac{\E[I(h(X, \widehat \beta) > c)\Pr[Y = 1|X, R=1]|S=0]}{\E[I(h(X, \widehat \beta) > c)|S=0]}
\end{equation}
or equivalently using the weighting representation
\begin{equation}
\label{id:ipw-ppv}
\frac{\E\left[\frac{\Pr[S=0|X]}{\Pr[R=1|X]}I(h(X, \widehat \beta) > c, Y=1, R=1) \right]}{\E[I(S=0)I(h(X, \widehat \beta) > c)]} 
\end{equation}
\end{theorem}

\begin{proof}
For the outcome model representation
\begin{align*}
\E[Y=1|I(h(X, \widehat \beta) > c), S=0] &= \frac{\E[I(h(X, \widehat \beta) > c, Y = 1)|S=0]}{\E[I(h(X, \widehat \beta) > c)|S=0]} \\
&= \frac{\E[\E[I(h(X, \widehat \beta) > c, Y = 1)|X, S=0]|S=0]}{\E[I(h(X, \widehat \beta) > c)|S=0]} \\
&= \frac{\E[\E[I(h(X, \widehat \beta) > c, Y = 1)|X, R=1]|S=0]}{\E[I(h(X, \widehat \beta) > c)|S=0]} \\
&= \frac{\E[I(h(X, \widehat \beta) > c)\Pr[Y = 1|X, R=1]|S=0]}{\E[I(h(X, \widehat \beta) > c)|S=0]} 
\end{align*}
For the inverse-odds weighting we have\allowdisplaybreaks
\begin{align*}
\E[Y=1|I(h(X, \widehat \beta) > c), S=0] &= \frac{\E[I(h(X, \widehat \beta) > c, Y = 1)|S=0]}{\E[I(h(X, \widehat \beta) > c)|S=0]} \\
&= \frac{\E[\E[I(h(X, \widehat \beta) > c, Y = 1)|X,S=0]|S=0]}{\E[I(h(X, \widehat \beta) > c)|S=0]} \\
&= \frac{\E[I(S=0)I(h(X, \widehat \beta) > c) \Pr[Y = 1|X,S=0]]}{\E[I(S=0) I(h(X, \widehat \beta) > c)]} \\
&= \frac{\E[I(S=0)I(h(X, \widehat \beta) > c) \Pr[Y = 1|X,R=1]]}{\E[I(S=0)I(h(X, \widehat \beta) > c)]} \\
&= \frac{\E\left[I(S=0)I(h(X, \widehat \beta) > c) \frac{\Pr[Y = 1,R=1|X]}{\Pr[R=1|X]}\right]}{\E[I(S=0)I(h(X, \widehat \beta) > c)]} \\
&= \frac{\E\left[\E\left[I(S=0)I(h(X, \widehat \beta) > c) \frac{\Pr[Y = 1,R=1|X]}{\Pr[R=1|X]}\Big|X\right]\right]}{\E[I(S=0)I(h(X, \widehat \beta) > c)]} \\
&= \frac{\E\left[\E\left[\frac{\Pr[R=0|X]}{\Pr[R=1|X]}I(h(X, \widehat \beta) > c) \Pr[Y = 1,R=1|X]\Big|X\right]\right]}{\E[I(S=0)I(h(X, \widehat \beta) > c)]} \\
&= \frac{\E\left[\E\left[\frac{\Pr[R=0|X]}{\Pr[R=1|X]}I(h(X, \widehat \beta) > c, Y=1, R=1) \Big|X\right]\right]}{\E[I(S=0)I(h(X, \widehat \beta) > c)]} \\
&= \frac{\E\left[\frac{\Pr[S=0|X]}{\Pr[R=1|X]}I(h(X, \widehat \beta) > c, Y=1, R=1) \right]}{\E[I(S=0)I(h(X, \widehat \beta) > c)]} 
\end{align*}
\end{proof}
\noindent The following theorem gives the identifiability result for the negative predictive value and the proof is similar to the proof of Theorem \ref{thm:id-ppv}.
\begin{theorem}
\label{thm:id-npv}
If assumptions A1, A2 and that $\E[I(h(X, \widehat \beta) \leq  c)|S=0]>0$ hold, then the negative predictive value in the target population is identifiable using the observed data through the observed data functional 
\begin{equation}
\label{id:out-ppv}
\frac{\E[I(h(X, \widehat \beta) \leq c)\Pr[Y = 0|X, R=1]|S=0]}{\E[I(h(X, \widehat \beta) \leq c)|S=0]}
\end{equation}
or equivalently using the weighting representation
\begin{equation}
\label{id:ipw-ppv}
\frac{\E\left[\frac{\Pr[S=0|X]}{\Pr[R=1|X]}I(h(X, \widehat \beta) \leq c, Y=0, R=1) \right]}{\E[I(S=0)I(h(X, \widehat \beta) \leq c)]}.
\end{equation}
\end{theorem}
\subsection{Estimation}

For specificity the sample analogs of the identifiability result in Theorem \ref{thm:id-spec} give the outcome model estimator
\begin{equation}
\label{spec-out}
\widehat \psi_{spec, out} = \frac{\sum_{i=1}^n I(R_i = 0) I(h(X_i, \widehat \beta) \leq c) (1 - \widehat m(X_i))}{\sum_{i=1}^n I(R_i=0)  (1 - \widehat m(X_i))},
\end{equation}
and the weighting estimator
\begin{equation}
\label{spec-iw}
\widehat \psi_{spec, w} =\frac{\sum_{i=1}^n I(h(X_i, \widehat \beta) \leq c, Y_i=0, R_i =1) \widehat w(X_i)}{\sum_{i=1}^n I(Y_i=0, R_i =1) \widehat w(X_i)}.
\end{equation}
Lastly, the doubly robust estimator for specificity $\widehat \psi_{spec, dr}$ is equal to
\[
\frac{\sum_{i=1}^n \left(I(S_i=0) I(h(X_i, \widehat \beta) \leq c) (1 - \widehat m(X_i)) +  \widehat w(X_i) I(R_i=1) I(h(X_i, \widehat \beta) \leq c) \big\{ I(Y_i = 0) - (1-\widehat m(X_i))\big\}\right)}{\sum_{i=1}^n \left(I(S_i=0) (1 - \widehat m(X_i)) +  \widehat w(X_i) I(R_i=1) \big\{ I(Y_i = 0) - (1 - \widehat m(X_i)) \big\}\right)}
\]
Using the sample analogs of the identifiability expressions in Theorem \ref{thm:id-ppv} gives the outcome model estimator for the positive predictive value in the target population
\begin{equation}
\label{ppv-out}
\widehat \psi_{ppv, out} = \frac{\sum_{i=1}^n I(R_i = 0) I(h(X_i, \widehat \beta) > c) \widehat m(X_i)}{\sum_{i=1}^n I(R_i=0)  I(h(X_i, \widehat \beta) > c)},
\end{equation}
and the weighting estimator
\begin{equation}
\label{sens-w}
\widehat \psi_{ppv, w} =\frac{\sum_{i=1}^n I(h(X_i, \widehat \beta) > c, Y_i=1, R_i =1) \widehat w(X_i)}{\sum_{i=1}^n I(R_i =1) \widehat w(X_i) I(h(X_i, \widehat \beta) > c)}.
\end{equation}
Similar calculations as for the doubly robust estimators of sensitivity result in the doubly robust estimator of the positive predictive value $\widehat \psi_{ppv, dr}$ given by the formula
\[
\frac{\sum_{i=1}^n \left(I(S_i=0) I(h(X_i, \widehat \beta) > c) \widehat m(X_i) +  \widehat w(X_i) I(R_i=1) I(h(X_i, \widehat \beta) > c) \big\{ I(Y_i = 1) - \widehat m(X_i)\big\}\right)}{\sum_{i=1}^n I(R_i=0)  I(h(X_i, \widehat \beta) > c)}
\]
Similarly we have the outcome model estimator for the negative predictive value in the target population 
\begin{equation}
\label{npv-out}
\widehat \psi_{npv, out} = \frac{\sum_{i=1}^n I(R_i = 0) I(h(X_i, \widehat \beta) \leq c) (1 - \widehat m(X_i))}{\sum_{i=1}^n I(R_i=0)  I(h(X_i, \widehat \beta) \leq c)},
\end{equation}
and the weighting estimator
\begin{equation}
\label{spec-iw}
\widehat \psi_{npv, w} =\frac{\sum_{i=1}^n I(h(X_i, \widehat \beta) \leq c, Y_i=0, R_i =1) \widehat w(X_i)}{\sum_{i=1}^n I(R_i =1) \widehat w(X_i) I(h(X_i, \widehat \beta) \leq c)}.
\end{equation}
The doubly robust estimator for the negative predictive value $\widehat \psi_{npv, dr}$ is given by 
\[
\frac{\sum_{i=1}^n \left(I(S_i=0) I(h(X_i, \widehat \beta) \leq c) (1 - \widehat m(X_i)) +  \widehat w(X_i) I(R_i=1) I(h(X_i, \widehat \beta) \leq c) \big\{ I(Y_i = 0) - (1-\widehat m(X_i))\big\}\right)}{\sum_{i=1}^n I(R_i=0)  I(h(X_i, \widehat \beta) \leq c)}
\]
The doubly robust estimators of specificity, positive and negative predictive value are consistent if at least one (but not necessarily both) of $\widehat m(X)$ and $\widehat w(X)$ are correctly specified.

For loss-based measures of model performance, the target parameter is the risk (expected loss) in the target population
\[
\E[L(Y, h(X, \widehat \beta))|S=0],
\]
which is identifiable through the expression
\begin{align*}
\E[L(Y, h(X, \widehat \beta))|S=0] &= \E[\E[L(Y, h(X, \widehat \beta))|X, S=0]|S=0] \\
&= \E[\E[L(Y, h(X, \widehat \beta))|X, R=1]|S=0]
\end{align*}
and
\begin{align*}
\E[\E[L(Y, h(X, \widehat \beta))|X, R=1]|S=0] &= \frac{1}{\Pr[S=0]}\E\left[I(R=0)\E\left[\frac{I(R=1)}{\Pr[R=1|X]} L(Y, h(X, \widehat \beta))\big |X\right]\right] \\
&= \frac{1}{\Pr[S=0]}\E\left[\E\left[\frac{I(R=1) \Pr[R=0|X]}{\Pr[R=1|X]} L(Y, h(X, \widehat \beta))\big |X\right]\right] \\
&= \frac{1}{\Pr[S=0]}\E\left[\frac{I(R=1) \Pr[R=0|X]}{\Pr[R=1|X]} L(Y, h(X, \widehat \beta))\right].
\end{align*}
The sample analogs of the identifibaiblity expressions above are the outcome model estimator
\[
\widehat \psi_{loss, out} = \frac{1}{n_0} \sum_{i=1}^n I(S_i = 0) \widehat \E[L(Y, h(X, \widehat \beta))|X_i, R_i=1]
\] 
where $\widehat \E[L(Y, h(X, \widehat \beta))|X, R=1]$ is an estimator for $\E[L(Y, h(X, \widehat \beta))|X, R=1]$. We also have the weighting estimator
\[
\widehat \psi_{loss, w} = \frac{1}{n_0} \sum_{i=1}^n I(R_i = 1) \widehat w(X _i) L(Y_i, h(X_i, \widehat \beta)).
\]
Following the arguments in \cite{morrison2023robust}, the doubly robust estimator is defined as
\[
\frac{1}{n_0} \sum_{i=1}^n \left[ I(R_i = 1) \widehat w(X _i) \left( L(Y, h(X, \widehat \beta)) - \widehat \E[L(Y, h(X, \widehat \beta))|X_i, R=1] \right) + I(S_i = 0) \widehat \E[L(Y, h(X, \widehat \beta))|X_i, R=1]\right].
\]

\section{Additional results under the exponential tilt model for sensitivity}
\label{app:sens}

Under the exponential tilt model we have 
\begin{equation*}
  \begin{split}
  f_{Y | X, S}( y | x, s = 0) 
      &= \dfrac{ e^{\gamma_{s'} y} f_{Y | X, S}(y| x, s = s') }{ \E [ e^{\gamma_{s'} y} | X = x, S = s']}.
  \end{split}
\end{equation*}
which implies that for a function $g(Y,X)$ 
\begin{align*}
\E[g(Y,X)|X, S=0] &= \int g(y,X) df(y|X,S=0) \\
& = \dfrac{ \int g(y,X) e^{\gamma_{s'} } df_{Y | X, S}(y| x, s = s') }{ \E [ e^{\gamma_{s'} Y} | X = x, S = s']} \\
& = \frac{\E[ g(Y,X) e^{\gamma_{s'} Y}|X, S=s']}{\E[e^{\gamma_{s'} Y}|X, S=s']}
\end{align*}
Using this result we have 
\[
\E[\E[I(h(X, \widehat \beta) > c, Y=1)|X, S=0]|S=0] = \E \left[\frac{\E[ I(h(X, \widehat \beta) > c, Y=1) e^{\gamma_{s'} Y}|X, S=s']}{\E[e^{\gamma_{s'} Y}|X, S=s']} \Bigg|S=0\right]
\]
and
\[
\E[\E[I(Y=1)|X, S=0]|S=0] = \E \left[\frac{\E[ I(Y=1) e^{\gamma_{s'} q(Y)}|X, S=s']}{\E[e^{\gamma_{s'} Y}|X, S=s']} \Bigg|S=0\right]
\]
Using this results the sensitivity in the target population under the exponential tilt model can be identified using 
\begin{align*}
\E[I(h(X, \widehat \beta) > c)|Y = 1, S=0] &= \frac{\E[I(h(X, \widehat \beta) > c, Y = 1)|S=0]}{\E[I(Y=1)|S=0]} \\
&= \frac{\E[\E[I(h(X, \widehat \beta) > c, Y = 1)|X, S=0]|S=0]}{\E[\E[I(Y=1)|X,S=0]|S=0]} \\
&= \frac{\E[\E[I(h(X, \widehat \beta) > c, Y = 1)|X, S=0]|S=0]}{\E[\E[I(Y=1)|X,R=1]|S=0]} \\
&= \frac{\E \left[\frac{\E[ I(h(X, \widehat \beta) > c, Y=1) e^{\gamma_{s'} Y}|X, S=s']}{\E[e^{\gamma_{s'} Y}|X, S=s']} \Bigg|S=0\right]}{\E \left[\frac{\E[ I(Y=1) e^{\gamma_{s'} Y}|X, S=s']}{\E[e^{\gamma_{s'} Y}|X, S=s']} \Bigg|S=0\right]}
\end{align*}

\section{Additional simulation results}
\label{app-sim}

Figure \ref{fig:cor-spec} show results from simulations were the model is correct and the relative bias $<2\%$ for both outcome model and weighting estimator and the estimator calculated using the data from the source population is biased ($4\% - 163\%$ relative bias).

Figures \ref{fig:gam-cor-spec} and \ref{fig:gam-mis-spec} shows the results when the outcome model and model for study participation are estimated using a generalized additive model when the model being estimated is correctly specified and when it is misspecified, respectively. The generalized additive models where fit using splines for the main effect of each covariate. The results are similar to when generalized linear models are used to estimate the outcome model and the model for study participation.

\begin{figure}[htbp]
    \centering
    \includegraphics[width = \textwidth]{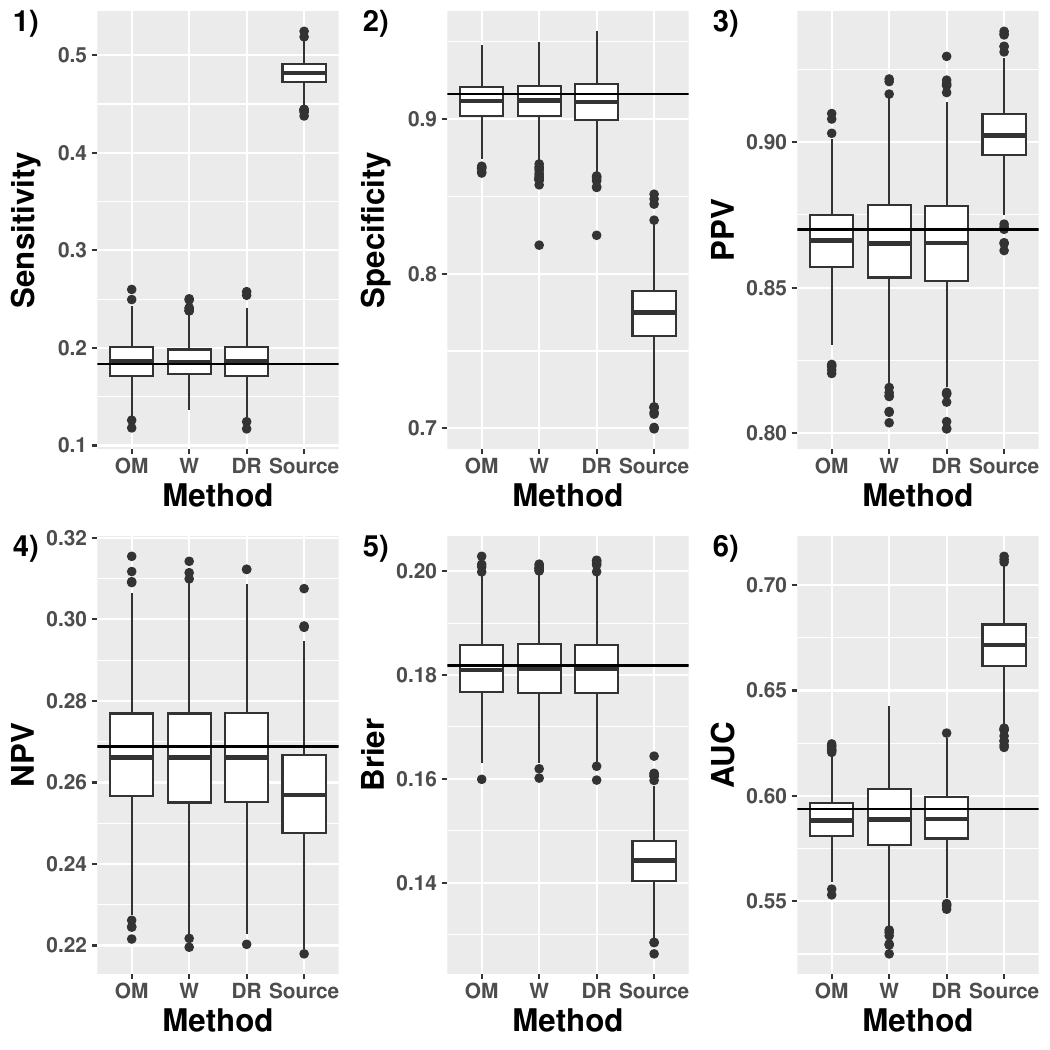}
    \caption{\label{fig:cor-spec} Simulation results for estimating model performance in the target population estimated using the outcome model (OM) estimator, weighting (W) estimator, doubly robust (DR) and an estimator that uses only data from the collection of studies (Source). The measures of model performance evaluated are: sensitivity (Plot 1), specificity (Plot 2), positive predictive value (Plot 3), negative predictive value (Plot 4), Brier score (Plot 5), and area under the curve (Plot 6). The horizontal line is the true measure of model performance (calculated analytically). Here, the model that is evaluated is correctly specified.}
\end{figure}

\begin{figure}[htbp]
    \centering
    \includegraphics[width = \textwidth]{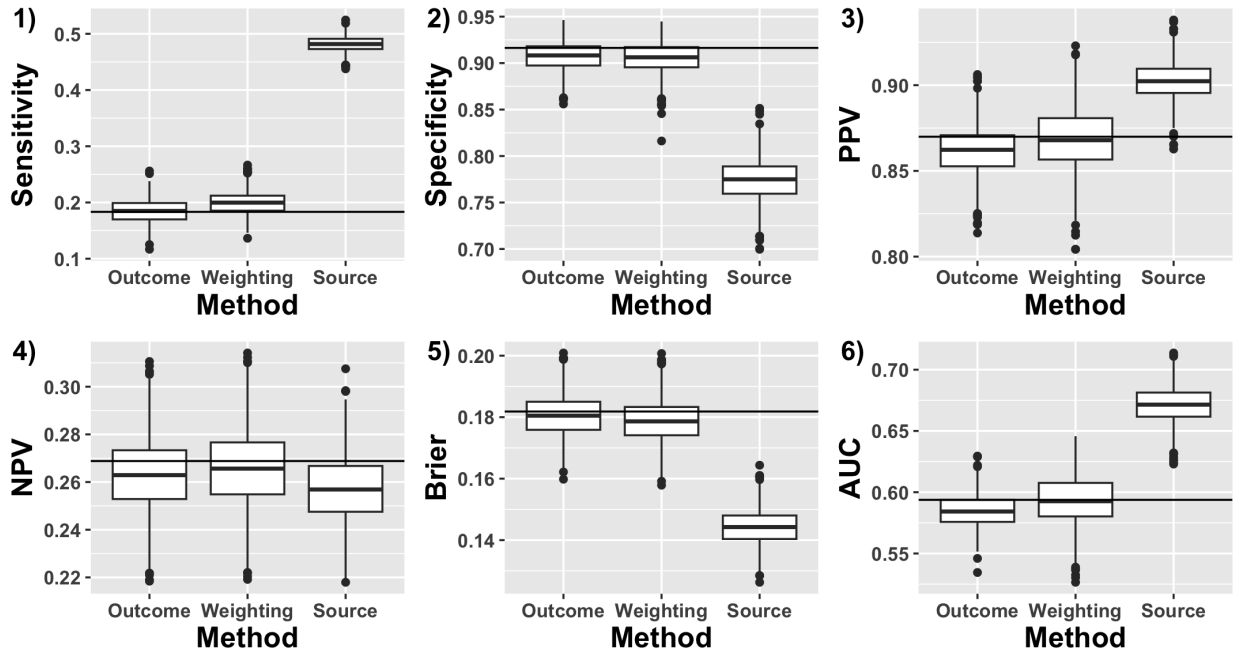}
    \caption{\label{fig:gam-cor-spec} Simulation results for estimating model performance in the target population estimated using the outcome model estimator, weighting estimator, and an estimator that uses only data from the collection of studies. For both the outcome model and the model for study participation are estimated using a generalized additive model. The measures of model performance evaluated are: sensitivity (Plot 1), specificity (Plot 2), positive predictive value (Plot 3), negative predictive value (Plot 4), Brier score (Plot 5), and area under the curve (Plot 6). The horizontal line is the true measure of model performance (calculated analytically). Here, the model that is evaluated is correctly specified.}
\end{figure}

\begin{figure}[htbp]
    \centering
    \includegraphics[width = \textwidth]{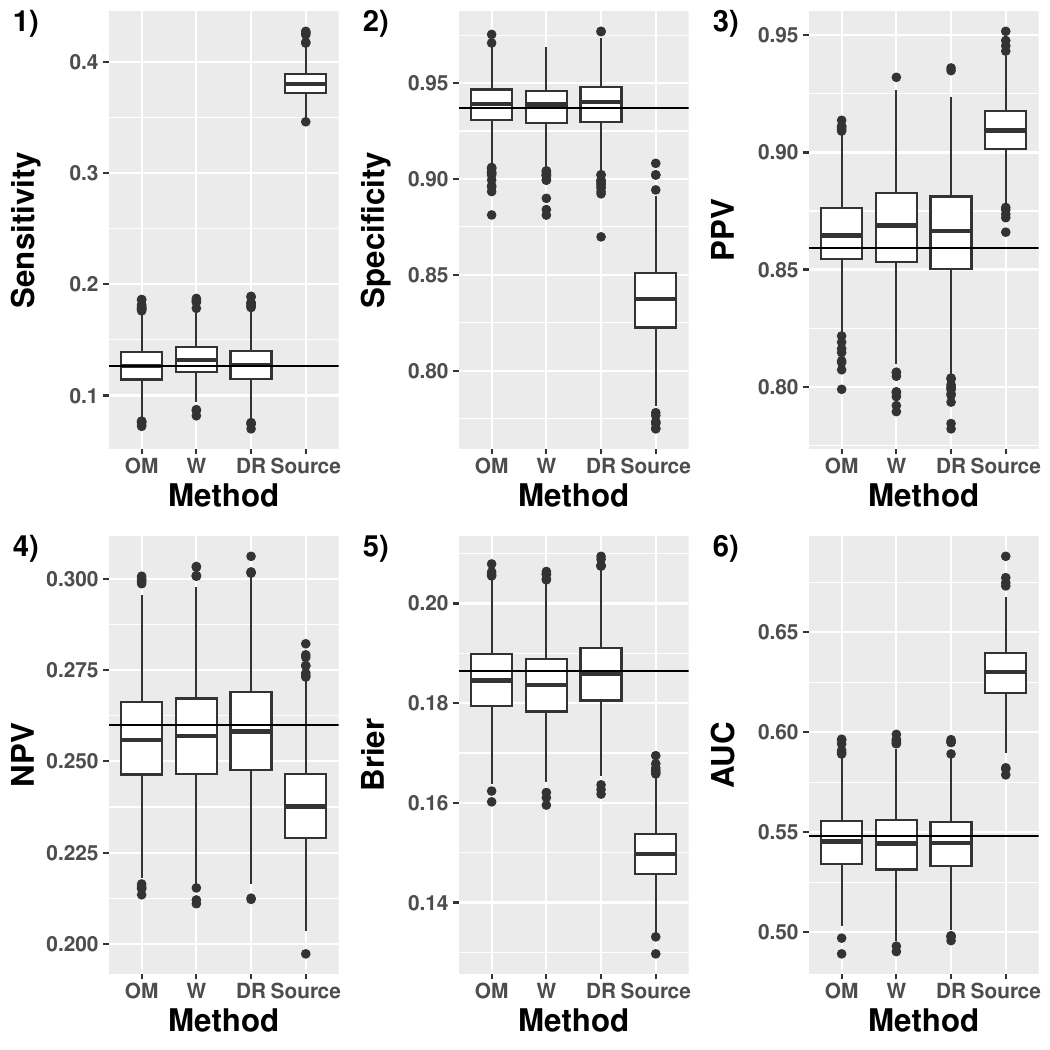}
    \caption{\label{fig:gam-mis-spec} Simulation results for estimating model performance in the target population estimated using the outcome model estimator, weighting estimator, and an estimator that uses only data from the collection of studies. For both the outcome model and the model for study participation are estimated using a generalized additive model. The measures of model performance evaluated are: sensitivity (Plot 1), specificity (Plot 2), positive predictive value (Plot 3), negative predictive value (Plot 4), Brier score (Plot 5), and area under the curve (Plot 6). The horizontal line is the true measure of model performance (calculated analytically). Here, the model that is evaluated is misspecified.}
\end{figure}

\section{Accounting for survey sampling weights in the sample from the target population}
\label{sec:weights}

In this section we show how the estimators for model performance in the target population can be extended to sampling designs where each observation in the target population is associated with a survey sampling weight $w$. 

The modified outcome model estimator for sensitivity in the target population is
\begin{equation*}
\frac{\sum_{i=1}^n I(R_i = 0) w_i I(h(X_i, \widehat \beta) > c) \widehat m(X_i)}{\sum_{i=1}^n I(R_i=0) w_i  \widehat m(X_i)}.
\end{equation*}
As $\widehat m(X)$ is an estimator for $\Pr[Y = 1|X, R=1]$ does not involve target population data it does not need to be modified.

The weighting estimator is given by
\begin{equation*}
\frac{\sum_{i=1}^n I(h(X_i, \widehat \beta) > c, Y_i=1, R_i =1) \tilde  w(X_i)}{\sum_{i=1}^n I(Y_i=1, R_i =1) \tilde w(X_i)},
\end{equation*}
where $\tilde  w(X)$ is an estimator for $\frac{\Pr[R=0|X]}{\Pr[R=1|X]}$ that incorporates the sampling weights (e.g., through a weighted logistic regression). The modified doubly robust estimator is 
\[
\frac{\sum_{i=1}^n \left(w_i I(S_i=0) I(h(X_i, \widehat \beta) > c) \widehat m(X_i) +  \tilde w(X_i) I(R_i=1) I(h(X_i, \widehat \beta) > c) \big\{ I(Y_i = 1) - \widehat m(X_i)\big\}\right)}{\sum_{i=1}^n \left(w_i I(S_i=0) \widehat m(X_i) +  \tilde w(X_i) I(R_i=1) \big\{ I(Y_i = 1) - \widehat m(X_i)\big\}\right)}.
\]
For specificity the modified outcome model estimator is
\begin{equation*}
\frac{\sum_{i=1}^n I(R_i = 0) w_i I(h(X_i, \widehat \beta) \leq c) (1 - \widehat m(X_i))}{\sum_{i=1}^n I(R_i=0) w_i  (1 - \widehat m(X_i))},
\end{equation*}
and the modified weighting estimator is given by
\begin{equation*}
\frac{\sum_{i=1}^n I(h(X_i, \widehat \beta) \leq c, Y_i=0, R_i =1) \tilde w(X_i)}{\sum_{i=1}^n I(Y_i=1, R_i =1) \tilde w(X_i)}.
\end{equation*}
The modified doubly robust estimator for specificity is given by
\[
\frac{\sum_{i=1}^n \left(w_i I(S_i=0) I(h(X_i, \widehat \beta) \leq c) (1 - \widehat m(X_i)) +  \tilde w(X_i) I(R_i=1) I(h(X_i, \widehat \beta) \leq c) \big\{ I(Y_i = 0) - (1-\widehat m(X_i))\big\}\right)}{\sum_{i=1}^n \left(w_i I(S_i=0) (1 - \widehat m(X_i)) +  \tilde w(X_i) I(R_i=1) \big\{ I(Y_i = 0) - (1 - \widehat m(X_i)) \big\}\right)}
\]
The modified outcome model estimator for the area under the curve in the target population is given by
\begin{equation*}
\frac{ \sum_{i \neq j} w_i w_j \widehat m(X_i) (1 -\widehat  m(X_j))   I(h( X_i,  \widehat \beta) > h( X_j, \widehat  \beta), S_i=0, S_j=0)}{\sum_{i \neq j} w_i w_j  \widehat m(X_i) (1 -\widehat  m(X_j)) I(S_i = 0, S_j = 0) },
\end{equation*}
and the modified weighting-based estimator for the AUC in the target population is given by
\begin{equation*}
\begin{split}
& \frac{\sum_{i\neq j} \tilde w(X_i) \tilde w(X_j) I(h( X_i, \widehat \beta) > h( X_j, \widehat \beta), Y_i=1, Y_j=0, R_i=1, R_j=1)}{ \sum_{i \neq j} \tilde w(X_i)\tilde w(X_j) I(Y_i=1, Y_j=0, R_i=1, R_j=1)}.
 \end{split}
\end{equation*}
and the modified doubly robust estimator is given by
\begin{align*}
d^{dr,w}(O_i, O_j; k(X_i,X_j)) &= d^{w,w}(O_i, O_j; k(X_i,X_j)) +  d^{out,w} (O_i, O_j; k(X_i,X_j)) -\\
&\tilde w(X_i) \widehat w(X_j) \widehat m(X_i) (1 - \widehat m(X_j))   I(R_i=1, R_j=1) k(X_i,X_j).
\end{align*}
where $d^{w,w}(O_i, O_j; k(X_i,X_j))$ and $d^{out,w} (O_i, O_j; k(X_i,X_j))$ are weighted versions of $d^{w}(O_i, O_j; k(X_i,X_j))$ and $d^{out} (O_i, O_j; k(X_i,X_j))$.
The modified outcome model estimator for the positive predictive value in the target population is given by
\begin{equation*}
\frac{\sum_{i=1}^n I(R_i = 0) w_i I(h(X_i, \widehat \beta) > c) \widehat m(X_i)}{\sum_{i=1}^n I(R_i=0) w_i I(h(X_i, \widehat \beta) > c)},
\end{equation*}
and the modified weighting estimator is
\begin{equation*}
\frac{\sum_{i=1}^n I(h(X_i, \widehat \beta) > c, Y_i=1, R_i =1) \tilde w(X_i)}{\sum_{i=1}^n I(R_i =1) \tilde w(X_i) I(h(X_i, \widehat \beta) > c)}.
\end{equation*}
The modified doubly robust estimator for the positive predictive value in the target population is
\[
\frac{\sum_{i=1}^n \left(w_i I(S_i=0) I(h(X_i, \widehat \beta) > c) \widehat m(X_i) +  \tilde w(X_i) I(R_i=1) I(h(X_i, \widehat \beta) > c) \big\{ I(Y_i = 1) - \widehat m(X_i)\big\}\right)}{\sum_{i=1}^n I(R_i=0) w_i I(h(X_i, \widehat \beta) > c)}
\]
The modified outcome model estimator for the negative predictive value in the target population is
\begin{equation*}
\frac{\sum_{i=1}^n w_i I(R_i = 0) I(h(X_i, \widehat \beta) \leq c) (1 - \widehat m(X_i))}{\sum_{i=1}^n w_i I(R_i=0)  I(h(X_i, \widehat \beta) \leq c)},
\end{equation*}
and the modified weighting estimator is given by
\begin{equation*}
\frac{\sum_{i=1}^n I(h(X_i, \widehat \beta) \leq c, Y_i=0, R_i =1) \tilde w(X_i)}{\sum_{i=1}^n I(R_i =1) \tilde w(X_i) I(h(X_i, \widehat \beta) \leq c)}.
\end{equation*}
The modified doubly robust estimator for the negative predictive value is
\[
\frac{\sum_{i=1}^n \left(w_i I(S_i=0) I(h(X_i, \widehat \beta) \leq c) (1 - \widehat m(X_i)) +  \tilde w(X_i) I(R_i=1) I(h(X_i, \widehat \beta) \leq c) \big\{ I(Y_i = 0) - (1-\widehat m(X_i))\big\}\right)}{\sum_{i=1}^n w_i I(R_i=0)  I(h(X_i, \widehat \beta) \leq c)}
\]
Lastly, the modified outcome model estimator for the Brier risk is given by
\[
\frac{1}{\sum_{i=1}^n I(R_i =0) w_i} \sum_{i=1}^n I(R_i = 0) w_i \widehat \E[L(Y, h(X, \widehat \beta))|X_i, R_i=1]
\] 
and the modified weighting estimator is given by
\[
\frac{1}{\sum_{i=1}^n I(R_i =0) w_i} \sum_{i=1}^n I(R_i = 1) \tilde w(X _i) L(Y_i, h(X_i, \widehat \beta)).
\]
The modified doubly robust estimator for the target population risk is given by 
\[
\frac{\sum_{i=1}^n \left[ I(R_i = 1) \tilde w(X _i) \left( L(Y, h(X, \widehat \beta)) - \widehat \E[L(Y, h(X, \widehat \beta))|X_i, R=1] \right) + w_iI(S_i = 0) \widehat \E[L(Y, h(X, \widehat \beta))|X_i, R=1]\right]}{\sum_{i=1}^n I(R_i =0) w_i}.
\]
\end{document}